\documentclass[12pt,english]{article}
\usepackage{amssymb,amsmath,amsthm}
\usepackage{amsfonts}
\usepackage{setspace,bbm}
\usepackage[longnamesfirst,authoryear]{natbib}
\usepackage{appendix}
\usepackage{enumerate}

\newtheorem{theorem}{Theorem}

\newtheorem{corollary}[theorem]{Corollary}

\newtheorem*{definition*}{Definition}

\newtheorem*{example*}{Example}

\newtheorem{lemma}[theorem]{Lemma}

\newtheorem{remark}{Remark}

\newcommand*\bTextWidth{\adjustbox{width=\textwidth,center}\bgroup}
\newcommand*\eTextWidth{\egroup}

\newcommand{\abs}[1]{\left\vert #1\right\vert}

\newcommand{\supp}[1]{\text{supp}\left(#1\right)}

\RequirePackage{graphicx}
\RequirePackage[normalem]{ulem}
\RequirePackage{xcolor}
\RequirePackage{amssymb}
\RequirePackage{amsmath}
\RequirePackage{amsfonts}
\RequirePackage{geometry}
\RequirePackage{xspace} %
\RequirePackage{verbatim}
\RequirePackage{environ}
\RequirePackage{bbm}
\RequirePackage{hyperref}

\newcommand*{\suppX}{\mathcal{S}_X}

\makeatletter
\renewcommand\subparagraph{\@startsection{subparagraph}{5}{\z@}%
                                     {-3.25ex\@plus -1ex \@minus -.2ex}%
                                     {0.0001pt \@plus .2ex}%
                                     {\normalfont\normalsize\bfseries}}
\makeatother

  \hypersetup{colorlinks=true,anchorcolor=blue,citecolor=black,linkcolor=[rgb]{0,0.0,0.75}} %

\makeatletter

\DeclareFontFamily{U}{mathc}{}
\DeclareFontShape{U}{mathc}{m}{it}%
{<->s*[1.03] mathc10}{}

\DeclareMathAlphabet{\mathscr}{U}{mathc}{m}{it}

\newcommand*\mmm{\mathscr{m}}

\input{named_thm_preamble}
\theoremstyle{namedAssumption}
\newtheorem{assumption}{Assumption}
\numberwithin{assumption}{section}
\numberwithin{example}{section}

\input{tcilatex}
\begin{document}

\title{\textsc{Nonparametric Identification and Estimation with
Non-Classical Errors-in-Variables}}
\author{Kirill S. \textsc{Evdokimov}\thanks{%
Universitat Pompeu Fabra and Barcelona School of Economics: \textsf{%
kirill.evdokimov@upf.edu}.} \and Andrei \textsc{Zeleneev}\thanks{%
University College London: \textsf{a.zeleneev@ucl.ac.uk}.} \thanks{%
Evdokimov gratefully acknowledges the support from the Spanish MCIN/AEI via
grants RYC2020-030623-I, PID2019-107352GB-I00, and Severo Ochoa Programme
CEX2019-000915-S.}}
\date{This version: December, 2023}%
\maketitle

\begin{abstract}
This paper considers nonparametric identification and estimation of the
regression function when a covariate is mismeasured. The measurement error
need not be classical. Employing the small measurement error approximation,
we establish nonparametric identification under weak and easy-to-interpret
conditions on the instrumental variable. The paper also provides
nonparametric estimators of the regression function and derives their rates
of convergence.

\end{abstract}

\newpage

\section{Introduction}

Regression is a fundamental tool for empirical analysis. Errors-in-Variables
(EIV) are a widespread problem in empirical applications. Mismeasurement of
a covariate, when not accounted for, may lead to biased estimates and
invalid inferences.

The goal of this paper is to study the nonparametric identification and
estimation of the regression function when a covariate is mismeasured.
Importantly, the measurement error need not be classical and can be
correlated with the mismeasured covariate. In this paper, we adopt the small
measurement error approximation, which allows us to provide a simple
nonparametric characterization of the problem. Then we provide transparent
and constructive identification analysis under weak and easy-to-interpret
conditions on the instrumental variable.

First, we focus on the Weakly Classical Measurement Error (WCME) model,
where the measurement error is uncorrelated with the true covariate but
generally is not independent from it. We show that the skedastic function of
the measurement errors, together with its derivative, plays a key role in
determining the bias of the naive regression estimator. We also show how the
EIV skedastic function can be recovered from the distribution of the
observables using a (possibly discrete) instrument, and how one can
construct a bias-corrected estimator of the regression function. We derive
its rate of convergence and provide conditions under which the approximation
error becomes negligible compared to the errors arising from the
nonparametric estimation of unknown functions in large samples.

Next, we consider the general Non-Classical Measurement Error (NCME) model,
which allows for a very broad form of EIV. In particular, the measurement
error can be correlated with the true covariate. Even though the NCME model
is much more general than the WCME model, we demonstrate how the results
from our analysis of the WCME model can be utilized to establish
identification of the general NCME model.

Importantly, our approach only requires an instrumental variable that can be
discrete. This allows for a broader range of applications compared to the
methods that require a continuously distributed instrument or the
availability of repeated (multiple) measurements of the true covariate.%
{}
In Section~\ref{sec:WCME}, we discuss in detail the exclusion and relevance
conditions that the instrumental variable needs to satisfy, and consider
some examples.

Our paper contributes to the large literature studying models with
mismeasured data. \cite{CarrollEtAl2006Book-ME,ChenHongNekipelov2011JEL};
and \cite{Schennach2020HB-ME,Schennach2022JEcLit} provide excellent
literature overviews.

Our main focus is on the settings where the distribution of measurement
error is unknown. Nonparametric analysis of the EIV problem in such settings
requires additional information to separate the true covariate from the
measurement error. In Economics, most commonly instrumental variables are
used for this purpose (%
\citealp{HINP1991JoE,HausmanNeweyPowell1995JoE,Newey2001REStat,Schennach2007Ecta,Hu2008JoE,HuSchennach2008Ecta,Wilhelm2019WP-TestingForME}%
, among others). Repeated measurements can also be utilized (%
\citealp{HINP1991JoE,LiVuong1998JoE,Schennach2004Ecta}, among others) but
are less frequently available. Note that repeated measurements can serve as
valid instruments in our analysis.

Small measurement error (SME) approximation has been widely employed in
Statistics and Econometrics to study the effect of EIV on various estimators
and to bias-correct them (e.g., %
\citealp{WolterFuller1982AS,CarrollStefanski1990JASA,Chesher1991Biomet,Chesher2000WP,CarrollEtAl2006Book-ME,ChesherSchluter2002ReStud}%
, among others). \cite{BoundBrownMathiowetz2001HBoE} document that the EIV
in economic applications are typically relatively small although are often
non-classical, which suggests that our analysis should be useful in many
applied settings. This paper differs from the previous literature in two
ways. First, it appears to be the first paper to study the nonparametric SME
approximation with non-classical EIV. Second, previously developed SME bias
reduction techniques usually assume that the EIV variance is either known or
can be directly estimated from an available dataset, e.g., using repeated
measurements. In contrast, this paper demonstrates how the whole EIV
skedastic function can be identified and estimated using only a (possibly
discrete) instrumental variable.

The analysis of this paper complements the existing \textquotedblleft
large\textquotedblright\ measurement error literature. By focusing on a
narrower range of settings, the paper provides simpler characterizations of
the problem and estimators, which are valid under very weak and
easy-to-interpret conditions on the instrumental variable. In particular,
our identification results do not rely on the completeness conditions, and
estimation does not involve solving ill-posed inverse problems or
deconvolution.

The rest of this paper is organized as follows. Section~\ref{sec:WCME}
studies the Weakly Classical Measurement Error (WCME) model. Section~\ref%
{sec:NCME} considers the general Non-Classical Measurement Error (NCME)
model. The proofs are collected in the Appendix.

\bigskip

\section{Weakly Classical Measurement Errors\label{sec:WCME}}

We consider the regression model%
\begin{equation}
\rho \left( x\right) \equiv E\left[ Y_{i}|X_{i}^{\ast }=x\right] ,
\label{eq:reg rho x def}
\end{equation}%
where $Y_{i}\in \mathbb{R}$ is the outcome variable, and $X_{i}^{\ast }\in 
\mathbb{R}$ is the true value of the covariate for individual $i$. The
researcher observed a mismeasured version of $X_{i}^{\ast }$:%
\begin{equation*}
X_{i}=X_{i}^{\ast }+\varepsilon _{i}.
\end{equation*}%
where $\varepsilon _{i}$ is the measurement error. The researcher has a
random sample of $\left( Y_{i},X_{i},Z_{i}\right) $, where $Z_{i}$ are
instrumental variables that are used to identify the model and will be
discussed later. It is straightforward to also include correctly measured
covariates into the model, see Remark~\ref{rem:covars W} for details.

In this section we consider the Weakly Classical Measurement Error (WCME)\
model:

\begin{assumption}[WCME]
\namedlabel{ass:WCME}{WCME} $X_{i}=X_{i}^{\ast }+\varepsilon _{i}$ and $E%
\left[ \varepsilon _{i}|X_{i}^{\ast }\right] =0$.
\end{assumption}

The measurement error $\varepsilon _{i}$ is uncorrelated with the true
covariate $X_{i}^{\ast }$. Assumption~\ref{ass:WCME} is significantly weaker
than the (Strongly) Classical Measurement Error (CME) assumption, since $%
\varepsilon _{i}$ need not be independent from $X_{i}^{\ast }$. For example,
the measurement error can be conditionally heteroskedastic, i.e., its
conditional variance 
\begin{equation*}
v\left( x \right) \equiv V\left[ \varepsilon _{i}|X_{i}^{\ast }=x\right] .
\end{equation*}
may depend on $x$. Function $v\left( x^{\ast }\right) $ is usually unknown.

\begin{example*}[WCME-LIN-RC]
Suppose $X_{i}=\psi _{i1}+\psi _{i2}X_{i}^{\ast }$, where $\left( \psi
_{i1},\psi _{i2}\right) \perp X_{i}^{\ast }$. Assumption~\ref{ass:WCME} is
satisfied if $E\left[ \left( \psi _{i1},\psi _{i2}\right) \right] =\left(
0,1\right) $. Here $v(x) =\sigma _{\psi _{1}}^{2}+\sigma _{\psi _{2}}^{2}
x^2 +2\sigma _{\psi _{1}\psi _{2}}x$.%
\end{example*}

In this paper we use the Small Measurement Error (SME)\ approximation (e.g., %
\citealp{WolterFuller1982AS}) for the analysis, i.e., we will consider the
approximations of the model when $v(x^*)$ and the higher conditional moments
of $\varepsilon_i$ are small.

Specifically, we model the measurement error as $\varepsilon _{i}=\tau \xi
_{i}$ where the distribution of $\xi _{i}$ is fixed, and $\tau $ is a
non-stochastic parameter. Assumption~\ref{ass:WCME} requires $E[\xi
_{i}|X_{i}^{\ast }]=0$. The conditional variance of $\varepsilon _{i}$ is
given by $v(x)=\tau ^{2}V[\xi _{i}|X_{i}^{\ast }=x]=O(\tau ^{2})$.

We study the properties of the model when $\tau \rightarrow 0$. Under some
smoothness conditions, 
\begin{equation*}
E\left[ Y_{i}|X_{i}=x\right] =\rho \left( x\right) +O\left( \tau ^{2}\right)
.
\end{equation*}%
Thus, a naive regression estimator of $\rho $ that ignores the presence of
the measurement errors in $X_{i}$ has a bias of order $O\left( \tau
^{2}\right) $, e.g., see \cite{Chesher1991Biomet}.

The goal of the small measurement error analysis is to provide a function $%
\widetilde{\rho }\left( x\right) $ that {}has a smaller bias, i.e., satisfies%
\begin{equation}
\widetilde{\rho }\left( x\right) =\rho \left( x\right) +O\left( \tau
^{p}\right) ,  \label{eq:intro:rho-tilde goal}
\end{equation}%
for some $p\geq 3$.

To identify the model we will rely on an observed instrumental variable
(instrument) $Z_{i}$ that satisfies the following exogeneity assumption.

\setcounter{assumption}{0}

\begin{assumption}
\label{ass:NPID:exclusion} $E\left[ Y_{i}|X_{i}^{\ast },Z_{i}\right] =E\left[
Y_{i}|X_{i}^{\ast }\right] $.
\end{assumption}

This assumption states that $Z_{i}$ is an \textquotedblleft
excluded\textquotedblright\ variable: given $X_{i}^{\ast }$, instrument $%
Z_{i}$ has no effect on the conditional mean of $Y_{i}$. Without loss of
generality, we can assume that $Z_{i}$ is discrete. (The instrument also
needs to satisfy a \textquotedblleft relevance\textquotedblright condition:
it needs to affect the conditional distribution $f_{X^{\ast }|Z}\left(
x|z\right) $. This condition will appear in Theorem~\ref{thm:NPID-non-cl:IV}%
.)

\begin{assumption}
\label{ass:NPID:Nondiffer ME} $E\left[ Y_{i}|X_{i}^{\ast },Z_{i},X_{i}\right]
=E\left[ Y_{i}|X_{i}^{\ast },Z_{i}\right] $.
\end{assumption}

Assumption~\ref{ass:NPID:Nondiffer ME} says that the measurement error $%
\varepsilon _{i}$ is nondifferential: conditional on $\left( X_{i}^{\ast
},Z_{i}\right) $, $X_{i}$ provides no additional information about (the
conditional mean of) $Y_{i}$. This assumption can be equivalently stated as $%
E\left[ Y_{i}|X_{i}^{\ast },Z_{i},\varepsilon _{i}\right] =E\left[
Y_{i}|X_{i}^{\ast },Z_{i}\right] =E\left[ Y_{i}|X_{i}^{\ast }\right] $. 

\begin{assumption}
\label{ass:NCME} $\varepsilon _{i}=\tau \xi _{i}$, where $\tau \geq 0$ is
non-random, $E[\xi _{i}|X_{i}^{\ast }]=0$, and $f_{\xi |X^{\ast }Z}\left(
u|x,z\right) =f_{\xi |X^{\ast }}\left( u|x\right) $ $\forall u,x,z$. 
{}
\end{assumption}

Assumption~\ref{ass:NCME} and the smoothness conditions below are stated
using the auxiliary variable $\xi _{i}$, whose variance does not shrink.
This allows formulating the smoothness conditions in the conventional form.
For example, we will assume that the density $f_{\xi |X^{\ast }}\left(
u|x\right) $ is bounded. In contrast, the density of $\varepsilon _{i}=\tau
\xi _{i}$ is $f_{\varepsilon |X^{\ast }}\left( e|x\right) =\frac{1}{\tau }%
f_{\xi |X^{\ast }}\left( \left. \frac{e}{\tau }\right\vert x\right) $ and is
not bounded as $\tau \rightarrow 0$. {} Assumption~\ref{ass:NCME} also implies that $\varepsilon
_{i}\perp Z_{i}|X_{i}^{\ast }$.

Finally, the following two assumptions are smoothness conditions.

\begin{assumption}
\label{ass:NPID:smoothness} Function $\rho \left( x\right) $ and the
conditional densities $f_{X^{\ast }|Z}(x|z)$ and $f_{\xi |X^{\ast }}\left( u
|x\right) $ are bounded functions with $m\geq p$ bounded derivatives with
respect to $x$, for some integer $p\geq 3$.
\end{assumption}

\begin{assumption}
\label{ass:NPID:dominance} $\int |u|^{m}\sup_{\tilde{x}\in \suppX}\left\vert
\nabla_{x}^{\ell }f_{\xi |X^{\ast }}(u|\tilde{x})\right\vert du< \infty$ for 
$\ell \in \{0,\ldots,m\}$ for some closed convex set $\suppX \subseteq 
\mathbb{R}$ containing the supports of $X_i^*$ and $X_i$.
\end{assumption}

Assumption \ref{ass:NPID:dominance} is a weak restriction imposed on the
conditional moments of $\xi _{i}$. 
Appendix \ref{sec: verification of dominance} provides a set of primitive
conditions that guarantee that Assumption \ref{ass:NPID:dominance} holds.
Also notice that Assumption \ref{ass:NPID:dominance} would automatically
hold if %
the support of $\xi _{i}$ is bounded, since $f_{\xi |X^{\ast }}(u |x )$ and
its derivatives are uniformly bounded under Assumption \ref%
{ass:NPID:smoothness}.

We can now state the first main result of the paper. Let%
\begin{eqnarray}
q\left( x,z\right) &\equiv &E\left[ Y_{i}|X_{i}=x,Z_{i}=z\right] ,\qquad
q\left( x\right) \equiv E\left[ Y_{i}|X_{i}=x\right] ,  \notag \\
s_{X|Z}\left( x|z\right) &\equiv &\frac{f_{X|Z}^{\prime }\left( x|z\right) }{%
f_{X|Z}\left( x|z\right) },\qquad s_{X^{\ast }|Z}\left( x|z\right) \equiv 
\frac{f_{X^{\ast }|Z}^{\prime }\left( x|z\right) }{f_{X^{\ast }|Z}\left(
x|z\right) },  \notag \\
\widetilde{v}\left( x\right) &\equiv &\frac{q\left( x,z_{1}\right) -q\left(
x,z_{2}\right) }{q^{\prime }\left( x\right) \left[ s_{X|Z}\left(
x|z_{1}\right) -s_{X|Z}\left( x|z_{2}\right) \right] },  \label{eq: v tilde}
\\
\widetilde{\rho }\left( x,z\right) &\equiv &q\left( x,z\right) -\widetilde{v}%
\left( x\right) \left[ q^{\prime }\left( x\right) s_{X|Z}\left( x|z\right) +%
\tfrac{1}{2}q^{\prime \prime }\left( x\right) \right] -q^{\prime }\left(
x\right) \widetilde{v}^{\prime }\left( x\right) .  \label{eq: rho tilde}
\end{eqnarray}%
Let $\mathcal{S}_{X^{\ast }}(z)$ denote the conditional support of $%
X_{i}^{\ast }|Z_{i}=z$. Consider any two values $z_{1}$ and $z_{2}$ the
instrument can take.

\begin{theorem}
\label{thm:NPID-non-cl:IV}Suppose that Assumptions~\ref{ass:WCME} and \ref%
{ass:NPID:exclusion}-\ref{ass:NPID:dominance} are satisfied. Suppose either
(i) $p=3$, or (ii) $E\left[ \xi _{i}^{3}|X_{i}^{\ast }\right] =0$ and $p=4$.
Consider any point $x\in \mathcal{S}_{X^{\ast }}(z_{1})\cap \mathcal{S}%
_{X^{\ast }}(z_{2})$ such that 
\begin{equation}
\rho ^{\prime }\left( x\right) \left[ s_{X^{\ast }|Z}\left( x|z_{1}\right)
-s_{X^{\ast }|Z}\left( x|z_{2}\right) \right] \neq 0.
\label{eq:thm:rank cond}
\end{equation}%
Then, as $\tau \rightarrow 0$,%
\begin{eqnarray*}
\widetilde{v}\left( x\right) &=&v\left( x\right) +O\left( \tau ^{p}\right)
,\quad \text{and} \\
\widetilde{\rho }\left( x,z_{1}\right) &=&\rho \left( x\right) +O\left( \tau
^{p}\right) .
\end{eqnarray*}
\end{theorem}

Theorem \ref{thm:NPID-non-cl:IV} demonstrates that $\widetilde \rho (x,z_1)$
identifies $\rho(x)$ up to an error of order $O(\tau^p)$ when $\tau
\rightarrow 0$. This is a substantial improvement over naive regression $%
q(x) $ which has a bias of order $O(\tau^2)$. The improvement in the
magnitude of the approximation error (from $O(\tau^2)$ to $O(\tau^4)$) is
especially noticeable when $E[\xi_i^3|X_i^*] = 0$, e.g., when the
measurement error is symmetric.

To establish the desired result, we first characterize the bias of $q(x,z)$
up to an error of order $O(\tau^p)$. The bias of $q(x,z)$ is of order $%
O(\tau^2)$ and determined by the conditional variance of the measurement
error $v(x)$ and its derivative $v^{\prime }(x)$, which are unknown. Then,
we show that $\widetilde v (x)$ identifies $v(x)$ up to an error of order $%
O(\tau^p)$.\footnote{%
We also demonstrate that $\widetilde v^{\prime }(x) = v^{\prime }(x) +
O(\tau^p)$. This is an important step of the proof.} This allows us to
approximate the bias of $q(x,z)$ with a sufficient precising using $%
\widetilde v(x)$ in place of $v(x)$. Finally, we construct $\widetilde \rho
(x,z_1)$ by bias correcting $q(x,z_1)$ and demonstrate that it approximates
the true regression function $\rho(x)$ up to an error of order $O(\tau^p)$.

The idea behind nonparametric identification is that although function $E%
\left[ Y_{i}|X_{i}^{\ast }=x,Z_{i}=z\right] $ does not depend on $z$,
function $q\left( x,z\right) \equiv E\left[ Y_{i}|X_{i}=x,Z_{i}=z\right] $
does vary with $z$. The theorem shows how this variation allows recovering $%
v\left( x\right) $. Specifically, the proof of the Theorem shows that%
\begin{equation}
q\left( x,z\right) =\rho \left( x\right) +v\left( x\right) \rho ^{\prime
}\left( x\right) s_{X^{\ast }|Z}\left( x|z\right) +\frac{1}{2}v\left(
x\right) \rho ^{\prime \prime }\left( x\right) +\rho ^{\prime }\left(
x\right) v^{\prime }\left( x\right) +O\left( \tau ^{p}\right) .
\label{eq:NPID:non-cl:key expr for q(x,z)}
\end{equation}%
Then it is shown that replacing the derivatives of $\rho $ with those of $q$
and $s_{X^{\ast }|Z}$ with $s_{X|Z}$ on the right-hand side in the above
equation does not increase the magnitude of the approximation error, i.e.,
that%
\begin{equation}
q\left( x,z\right) =\rho \left( x\right) +v\left( x\right) q^{\prime }\left(
x\right) s_{X|Z}\left( x|z\right) +\frac{1}{2}v\left( x\right) q^{\prime
\prime }\left( x\right) +q^{\prime }\left( x\right) v^{\prime }\left(
x\right) +O\left( \tau ^{p}\right) .
\label{eq:NPID:non-cl:key expr for q(x,z) - feasible}
\end{equation}

Note that only the second term on the right-hand side depends on $z$. Since $%
q$, $q^{\prime }$, and $s_{X|Z}$ are directly identified from the joint
distribution of the observables, considering the differences $q\left(
x,z_{1}\right) -q\left( x,z_{2}\right) $ then allows identification of $%
v\left( x\right) $ by $\widetilde{v}\left( x\right) $ up to an error of order%
$\ O\left( \tau ^{p}\right) $. This identification approach requires the
rank condition that $s_{X|Z}\left( x|z\right) $ depends on $z$, which is
ensured by equation~(\ref{eq:thm:rank cond}). In addition, it is necessary
that $q^{\prime }\left( x\right) \neq 0$, which is also ensured by equation~(%
\ref{eq:thm:rank cond}). The latter condition is weak: $\rho ^{\prime
}\left( x\right) =0$ for all $x$ only if $\rho \left( x\right) $ is a
constant.\footnote{%
For an analysis of the role of the conditions such as $\rho ^{\prime }\left(
x\right) \neq 0$ in the measurement error literature see \cite%
{EvdokimovZeleneev2018WP-Inference}.} {}

\begin{remark}
It is easy to check that $\widetilde{\rho }\left( x,z_{1}\right) =\widetilde{%
\rho }\left( x,z_{2}\right) $.{}
\end{remark}

\begin{corollary}[Classical Measurement Error]
\label{cor:NPID-CME}Suppose the hypotheses of Theorem~\ref%
{thm:NPID-non-cl:IV} hold, and the measurement error is classical, i.e., $%
\varepsilon _{i}\perp \left( X_{i}^{\ast },Z_{i},Y_{i}\right) $. Suppose
condition~(\ref{eq:thm:rank cond}) holds for some point $\dot{x}$. Then for
all $x$ and $z$ such that $x \in \mathcal{S}_{\mathcal{X}^*}(z)$, as $\tau
\rightarrow 0$, 
\begin{equation*}
\widetilde{\rho }_{\text{CME}}\left( x,z\right) =\rho \left( x\right)
+O\left( \tau ^{p}\right),
\end{equation*}%
where%
\begin{equation*}
\widetilde{\rho }_{\text{CME}}\left( x,z\right) \equiv q\left( x,z\right) -%
\widetilde{v}\left( \dot{x}\right) \left[ q^{\prime }\left( x\right)
s_{X|Z}\left( x|z\right) +\tfrac{1}{2}q^{\prime \prime }\left( x\right) %
\right] .
\end{equation*}
\end{corollary}

When the measurement error is classical, $v\left( x\right) $ is constant,
and hence the term containing $v^{\prime }\left( x\right) $ is absent from
equation~(\ref{eq:NPID:non-cl:key expr for q(x,z)}). In addition, Corollary~%
\ref{cor:NPID-CME} requires only a single point $\dot{x}$ satisfying the
rank condition (\ref{eq:thm:rank cond}) for identification of $\rho \left(
x\right) $ for all $x$, since $v\left( x\right) =v\left( \dot{x}\right) $
for all $x$.\footnote{%
For the result of Corollary \ref{cor:NPID-CME} to hold, it is sufficient to
require $v(x)$ to be constant, i.e., to assume that $\varepsilon_i$ is
homoskedastic, instead of requiring $\varepsilon _{i}\perp \left(
X_{i}^{\ast },Z_{i},Y_{i}\right)$.} In contrast, in the general case of
Theorem~\ref{thm:NPID-non-cl:IV}, nonparametric identification of $v\left(
x\right) $ and $\rho \left( x\right) $ for a given $x$ requires condition~(%
\ref{eq:thm:rank cond}) to hold at that point $x$. Identification of the
Classical Measurement Error model has been previously established in \cite%
{EvdokimovZeleneev-Estimation-WP-2022}.{}

\begin{remark}[Examples of IVs]
{} First, variable $X_{i}^{\ast }$ can be
caused by $Z_{i}$; for example, $X_{i}^{\ast }=q\left( Z_{i},\eta
_{i}\right) $ for some unobserved (vector) $\eta _{i}$ and function $q$.
Assumptions \ref{ass:NPID:exclusion} and \ref{ass:NPID:Nondiffer ME} will be
satisfied if $E\left[ U_{i}|Z_{i},\eta _{i},\varepsilon _{i}\right] =0$.

Second, variable $Z_{i}$ can be caused by $X_{i}^{\ast }$, for example be a
second measurement or proxy for $X_{i}^{\ast }$: $Z_{i}=\chi \left(
X_{i}^{\ast },\nu _{i}\right) $. For example, $Z_{i}$ can be a second
measurement: $Z_{i}=\alpha _{1}+\alpha _{2}X_{i}^{\ast }+\nu _{i}$.
Assumptions \ref{ass:NPID:exclusion} and \ref{ass:NPID:Nondiffer ME} will be
satisfied if $E\left[ U_{i}|X_{i}^{\ast },\nu _{i},\varepsilon _{i}\right]
=0 $.{}
\end{remark}

\begin{remark}
If the skedastic function $v\left( x\right) $ is known, there is no need in
having the instrumental variable $Z_{i}$. In this case one can use $%
\widetilde{\rho }\left( x\right) $ from equation~(\ref{eq: rho tilde}) with $%
q\left( x,z\right) $ and $\widetilde{v}\left( x\right) $ replaced by $%
q\left( x\right) $ and $v\left( x\right) $, and the conclusion of Theorem~%
\ref{thm:NPID-non-cl:IV} will continue to hold, i.e., $\widetilde{\rho }%
\left( x\right) =\rho \left( x\right) +O\left( \tau ^{p}\right) $.
\end{remark}

\begin{remark}
\label{rem:covars W}It is straightforward to include additional correctly
measured covariates $W_{i}$ into the model, and to consider regression
function $\rho \left( x,w\right) \equiv E\left[ Y_{i}|X_{i}^{\ast }=x,W_{i}=w%
\right] $. The correctly measured covariates $W_{i}$ play no special role,
and all of the analysis can be thought of as applying conditionally on $%
W_{i}=w$ for any given $w$, i.e., for the stratum with $W_{i}=w$. Thus, we
omit $W_{i}$ for simplicity of exposition.
\end{remark}

\bigskip

\paragraph{Nonparametric Estimation}

Theorem \ref{thm:NPID-non-cl:IV} suggests an analogue estimator of $%
\widetilde{\rho }$ by replacing functions that appear in equations~(\ref{eq:
v tilde})-(\ref{eq: rho tilde}) with their standard nonparametric estimators
(e.g., kernel or sieve), with optimally chosen tuning parameters. Let $%
\widehat{\rho }\left( x\right) $ denote this estimator. As an alternative to 
$\widehat{\rho }\left( x\right)$, we also consider $\widehat{\rho }^{\text{%
Naive}}\left( x\right)$, a naive nonparametric estimator of $\rho(x)$
ignoring the presence of the measurement error.

To approximate the finite sample properties of the studied estimators when $%
\tau$ is small, we consider a triangular asymptotic framework with drifting $%
\tau = \tau_n$ converging to zero as the sample size $n \rightarrow \infty$.

\begin{lemma}
\label{lem:NCME:IV:rates}Suppose the hypotheses of Theorem~\ref%
{thm:NPID-non-cl:IV} hold and $Z_{i}$ is discrete. Also, suppose $\tau
_{n}=o(1)$, then%
\begin{eqnarray*}
\widehat{\rho }\left( x\right) -\rho \left( x\right) &=&O_{p}\left( n^{-%
\frac{m-1}{2m+1}}+\tau _{n}^{p}\right) , \\
\widehat{\rho }^{\text{Naive}}\left( x\right) -\rho \left( x\right)
&=&O_{p}\left( n^{-\frac{m}{2m+1}}+\tau _{n}^{2}\right) .
\end{eqnarray*}
\end{lemma}

Lemma \ref{lem:NCME:IV:rates} establishes the rates of convergence for the
proposed and naive estimators. For each estimator, the rate of convergence
is determined by two components: the standard nonparametric learning rate
and the EIV (errors-in-variables) bias due to the presence of the
measurement error.

The nonparametric learning rate for $\widehat \rho (x)$ is slower than for
the naive estimator because it involves nonparametric estimation of
derivatives such as $q^{\prime }(x)$ and $s_{X|Z}(x|z)$. However, as Theorem %
\ref{thm:NPID-non-cl:IV} suggests, the EIV bias of the proposed estimator $%
\widehat \rho (x)$ is of order $O(\tau_n^p)$, whereas the naive estimator
has a much larger bias of order $O(\tau_n^2)$. Thus, despite the slower
nonparametric learning rate, $\widehat \rho (x)$ has a faster rate of
convergence than $\widehat{\rho }^{\text{Naive}}\left( x\right)$ unless $%
\tau_n$ is very small, i.e., the measurement error is negligible.

To illustrate this result, suppose the conditions of Theorem~\ref%
{thm:NPID-non-cl:IV}(ii) hold, $m=p=4$, and $\tau _{n}={}O\left( n^{-\frac{1%
}{12}}\right) $. Then $\widehat{\rho }\left( x\right) -\rho \left( x\right)
=O_{p}\left( n^{-\frac{1}{3}}\right) $, but the naive estimator has a much
slower rate of convergence: $\widehat{\rho }^{\text{Naive}}\left( x\right)
-\rho \left( x\right) =O_{p}\left( n^{-\frac{1}{6}}\right) $, because of the
EIV bias.

\bigskip

\bigskip

\section{General Non-Classical Measurement Error\label{sec:NCME}}

In this section, we will use notation $\mathcal{X}_{i}^{\ast }$ for the true
mismeasured covariate and consider the general measurement model 
\begin{equation}
X_{i}= \mmm \left( \mathcal{X}_{i}^{\ast },\psi _{i}\right) ,
\label{eq:NC ME def 1}
\end{equation}
where $\mmm$ is an unknown function and $\psi _{i}$ is a random vector
independent from $\left( Y_{i},\mathcal{X}_{i}^{\ast },Z_{i}\right) $.
Function $\mmm$ need not be monotone in any of the arguments. The
measurement error is non-classical: $X_{i}- \mathcal{X}_{i}^{\ast }$ and $%
\mathcal{X}_{i}^{\ast }$ are generally correlated.

As before, we want to identify and estimate the regression function 
\begin{equation}
\rho _{\mathcal{X}^{\ast }}\left( \varkappa \right) \equiv E\left[ Y_{i}|%
\mathcal{X}_{i}^{\ast }=\varkappa \right] .  \label{eq: NC ME rho def}
\end{equation}

We will assume that the measurement is sufficiently informative about the
true $\mathcal{X}_{i}^{\ast }$. Define%
\begin{equation*}
\mu \left( \varkappa ^{\ast }\right) \equiv E\left[ \left. X_{i}\right\vert 
\mathcal{X}_{i}^{\ast }=\varkappa ^{\ast }\right].
\end{equation*}

\begin{assumption}[MONOT-MEAS]
\namedlabel{ass:MONOT-MEAS}{MONOT-MEAS} $\mu \left( \varkappa ^{\ast
}\right) $ is a strictly increasing function.
\end{assumption}

\begin{example*}[NCME-LIN-RC, continued]
For $X_{i}=\psi _{i1}+\psi _{i2}\mathcal{X}_{i}^{\ast }$, we have $\mu
\left( \varkappa ^{\ast }\right) =c_{\psi 1}+c_{\psi 2}\varkappa ^{\ast }$,
and Assumption~\ref{ass:MONOT-MEAS} is satisfied if $c_{\psi 2}>0$. Note
that $\psi _{i2}$ is allowed to take negative values, which makes $\mmm$ a
decreasing function of $\mathcal{X}_{i}^{\ast }$ for such observations.
\end{example*}

Note that since functions $\mmm$ and $\rho _{\mathcal{X}^{\ast }}$ are
unrestricted, the model \eqref{eq:NC ME def 1}-\eqref{eq: NC ME rho def}
cannot be identified without some normalization or additional information.
Specifically, for any strictly increasing function $\lambda $, we can define
an observationally equivalent model with $\mathcal{\tilde{X}}_{i}^{\ast
}\equiv \lambda \left( \mathcal{X}_{i}^{\ast }\right) $, $\rho _{\mathcal{%
\tilde{X}}^{\ast }}\left( \tilde{\varkappa}^{\ast }\right) \equiv \rho _{%
\mathcal{X}^{\ast }}\left( \lambda ^{-1}\left( \tilde{\varkappa}^{\ast
}\right) \right) $ and $\mmm_{\mathcal{\tilde{X}}_{i}^{\ast }}\left( \tilde{%
\varkappa}^{\ast },\psi \right) \equiv \mmm\left( \lambda ^{-1}\left( \tilde{%
\varkappa}^{\ast }\right) ,\psi \right) $. {}

Let us \emph{define} random variable 
\begin{equation}
X_{i}^{\ast }\equiv \mu \left( \mathcal{X}_{i}^{\ast }\right) .
\label{eq:def:Xis}
\end{equation}%
Then,%
\begin{equation*}
E\left[ X_{i}|X_{i}^{\ast }\right] =E\left[ X_{i}|\mu \left( \mathcal{X}%
_{i}^{\ast }\right) \right] =E\left[ X_{i}|\mathcal{X}_{i}^{\ast }\right]
=\mu \left( \mathcal{X}_{i}^{\ast }\right) =X_{i}^{\ast },
\end{equation*}%
where the first equality follows by~(\ref{eq:def:Xis}), the second equality
follows from the strict monotonicity of $\mu (\cdot)$, the third equality is
the definition of $\mu \left( \cdot \right) $, and the last equality follows
from~(\ref{eq:def:Xis}).

Thus, for the general measurement error model we can consider an
observationally equivalent model that defines $X_{i}^{\ast }$ as in
equation~(\ref{eq:def:Xis}):%
\begin{eqnarray*}
\rho _{X^{\ast }}\left( x\right) &\equiv &E\left[ Y_{i}|X_{i}^{\ast }=x%
\right] , \\
X_{i} &=&X_{i}^{\ast }+\varepsilon _{i},\qquad E[\varepsilon
_{i}|X_{i}^{\ast }]=0.
\end{eqnarray*}%
Since in this model Assumption~\ref{ass:WCME} holds, we can apply the result
of Theorem~\ref{thm:NPID-non-cl:IV} to identify $\rho _{X^{\ast }}\left(
x\right) $ and $v\left( x\right) \equiv E\left[ \varepsilon
_{i}^{2}|X_{i}^{\ast }=x\right] $ (up to an error of order $O(\tau ^{p})$)
using $\widetilde{v}(x)$ and $\widetilde{\rho }(x)$ defined in equations~%
\eqref{eq: v tilde} and \eqref{eq: rho tilde}, respectively. Notice that $%
\rho _{\mathcal{X}^{\ast }}\left( \varkappa \right) =\rho _{X^{\ast }}\left(
\mu \left( \varkappa \right) \right) $. However, since $\mathcal{X}%
_{i}^{\ast }$ is not observed, one cannot identify $\mu \left( \varkappa
\right) $ and $\rho _{\mathcal{X}^{\ast }}(\varkappa )$ without some
sadditional information.

Suppose for a moment that the marginal distribution $F_{\mathcal{X}^{\ast }}$
of $\mathcal{X}_{i}^{\ast }$ is known (for example, from a separate dataset,
e.g., administrative records). In this case, we can identify $\rho _{%
\mathcal{X}^{\ast }}(\varkappa )$ up to an error of order $O(\tau ^{p})$
using 
\begin{equation}
\widetilde{\rho }_{\mathcal{X}^{\ast }}(\varkappa ,z)\equiv \widetilde{\rho }%
_{X^{\ast }}\left( \widetilde{Q}_{X^{\ast }}(F_{\mathcal{X}^{\ast
}}(\varkappa )),z\right) ,  \label{eq:def:rho tilde NC}
\end{equation}%
where 
\begin{equation}
\widetilde{Q}_{X^{\ast }}\left( s\right) \equiv Q_{X}\left( s\right) +\frac{1%
}{2}\left\{ s_{X}\left( Q_{X}\left( s\right) \right) \widetilde{v}\left(
Q_{X}\left( s\right) \right) +\nabla _{x}\widetilde{v}\left( Q_{X}\left(
s\right) \right) \right\} .  \label{eq:def:Q tilde NC}
\end{equation}%
Here $F_{\mathcal{X}^{\ast }}(\cdot )$ denotes the CDF of $\mathcal{X}%
_{i}^{\ast }$, and $Q_{X}(\cdot )$ denotes the quantile function (QF) of $%
X_{i}$. Note that all functions on the right-hand side of equation~(\ref%
{eq:def:Q tilde NC}) are identified directly from the observed data.

First, we demonstrate that $\widetilde Q_{X^*} (s) = Q_{X^*}(s) + O(\tau^p)$%
, where $Q_{X^*}(\cdot)$ denotes the quantile function of $X_i^*$. Combining
this with the result of Theorem \ref{thm:NPID-non-cl:IV} allows us to
establish the desired result formalized by the theorem below.

\setcounter{assumption}{0}

\begin{assumption}
\label{ass:NPID:integrable derivatives} $\int \left \vert \nabla_x^\ell
f_{X^*}(x) \right \vert dx < \infty$ for $\ell \in \{1, \ldots, p\}$.
\end{assumption}

\begin{theorem}
\label{thm:NCME rho tilde} Suppose that the hypotheses of Theorem \ref%
{thm:NPID-non-cl:IV} are satisfied for $x = Q_{X^* } \left(F_{\mathcal{X}^*}
(\varkappa)\right)$. Also, suppose Assumptions \ref{ass:MONOT-MEAS} and \ref%
{ass:NPID:integrable derivatives} hold. Then, as $\tau \rightarrow 0$, 
\begin{align*}
\widetilde \rho_{\mathcal{X}^*} (\varkappa, z_1) = \rho_{\mathcal{X}%
^*}(\varkappa) + O(\tau^p).
\end{align*}
\end{theorem}

Theorem \ref{thm:NCME rho tilde} demonstrates that, if the marginal
distribution of $\mathcal{X}_{i}^{\ast }$ is given, it is possible to
identify $\rho (\varkappa )$ up to an error of order $O(\tau ^{p})$ in the
general NCME model \eqref{eq:NC ME def 1} building on the identification
results for the WCME model.

Note that obtaining (an estimate of) the marginal distribution $F_{\mathcal{X%
}^{\ast }}$ is a much simpler task than obtaining a validation sample, i.e.,
the data on $\left( X_{i},\mathcal{X}_{i}^{\ast }\right) $ jointly. For
example, suppose $\mathcal{X}_{i}^{\ast }$ are individual wages, and $X_{i}$
are self-reported wages in a survey. The marginal distribution $F_{\mathcal{X%
}^{\ast }}$ can be provided by the Social Security Administration or similar
tax authorities in other countries. Providing such marginal distribution
does not pose any privacy risks. In contrast, obtaining a validation sample
that links individual's responses $X_{i}$ to the individual's social
security records $\mathcal{X}_{i}^{\ast }$ is a difficult task that in
particular faces major challenges concerning privacy.

If the distribution of $\mathcal{X}_{i}^{\ast }$ is unknown we can still
apply Theorem \ref{thm:NCME rho tilde} at $\varkappa =Q_{\mathcal{X}^{\ast
}}(q)$ for any quantile $q\in \left( 0,1\right) $ to identify $E[Y_{i}|%
\mathcal{X}_{i}^{\ast }=Q_{\mathcal{X}^{\ast }}(q)]=\rho _{\mathcal{X}^{\ast
}}(Q_{\mathcal{X}^{\ast }}(q))$, where $Q_{\mathcal{X}^{\ast }}(\cdot )$ is
the (unknown) quantile function of $\mathcal{X}_{i}^{\ast }$:%

\begin{corollary}
\label{cor: NCME quantiles} Suppose that the hypotheses of Theorem \ref%
{thm:NCME rho tilde} are satisfied for $x = Q_{X^*} (q)$. Then, as $\tau
\rightarrow 0$, 
\begin{align*}
\widetilde \rho_{X^*} \left(\widetilde Q_{X^*} (q), z_1\right) = \widetilde
\rho_{\mathcal{X}^*} (Q_{\mathcal{X}^*} (q), z_1) = \rho_{\mathcal{X}^*} (Q_{%
\mathcal{X}^*} (q)) + O (\tau^p).
\end{align*}
\end{corollary}

Corollary \ref{cor: NCME quantiles} demonstrates that even if $F_{\mathcal{X}%
^*}$ is unknown, we can still identify the conditional expectation of $Y_i^*$
given the $q$'th quantile of $\mathcal{X}_{i}^{\ast }$. Notice that in some
applications, the unobserved variable $\mathcal{X}_i^*$, for example an
individual's ability, might not even have well-defined economic units. In
such settings, identification of $E[Y_i | \mathcal{X}_i^* = Q_{\mathcal{X}%
^*} (q)]$ is fully exhaustive.

\bigskip

\bigskip

\bibliographystyle{ecta}
\bibliography{library}

\clearpage

\begin{appendices}%

\numberwithin{equation}{section} \numberwithin{remark}{section} %
\numberwithin{theorem}{section}

\section{\label{sec:Proofs}Proofs}

\subsection{Proof of Theorem \ref{thm:NPID-non-cl:IV}}

Before proving Theorem \ref{thm:NPID-non-cl:IV}, we state and prove 3 auxiliary lemmas. Then we prove the main result.

\subsubsection{Auxiliary lemmas}
\begin{lemma}
	\label{lem:q(x,z) expansion}
	Suppose that the hypotheses of Theorem \ref{thm:NPID-non-cl:IV} are satisfied. Then, for any $x$ and $z$ such that $x \in \mathcal S_{X^*} (z)$ (so, $f_{X^*|Z}(x|z) > 0$ and $f_{X^*}(x) > 0$), we have:
	\begin{enumerate}[(i)]
		\item \label{item: q(x,z) expansion}
		\begin{align}
			q\left( x,z\right) =\rho \left( x\right) +v\left( x\right) \left(\rho ^{\prime
					}\left( x\right) s_{X^{\ast }|Z}\left( x|z\right) +\frac{1}{2} \rho ^{\prime \prime }\left( x\right)\right) +\rho ^{\prime }\left(
			x\right) v^{\prime }\left( x\right) +O\left( \tau ^{p}\right), 	\label{eq:NPID:q(x,z) K2 exp}
		\end{align}
		\item \label{item: s q prime bias}
		\begin{align}
			\label{eq:pf: s q prime bias}
			\abs{s_{X|Z}(x|z) - s_{X^*|Z}(x|z)} + \abs{q'(x) - \rho'(x)} = O (\tau^2).
		\end{align}
	\end{enumerate}
\end{lemma}

\begin{proof}[Proof of Lemma \ref{lem:q(x,z) expansion}, Part (\ref{item: q(x,z) expansion})]
	For concreteness, we establish the result for the case $p = 4$. The proof for the case $p=3$ is analogous (but simpler).

	We have%
	\begin{eqnarray*}
	f_{X^{\ast }X|Z}\left( r,x|z\right) &=&f_{\varepsilon |X^{\ast }}\left(
	x-r|r\right) f_{X^{\ast }|Z}\left( r|z\right), \\
	f_{X^{\ast }|XZ}\left( r|x,z\right) &=&\frac{f_{X^{\ast }X|Z}\left(
	r,x|z\right) }{f_{X|Z}\left( x|z\right) }=\frac{f_{\varepsilon |X^{\ast
	}}\left( x-r|r\right) f_{X^{\ast }|Z}\left( r|z\right) }{f_{X|Z}\left(
	x|z\right) }, \\
	f_{X|Z}\left( x|z\right) &=&\int f_{\varepsilon |X^{\ast }}\left(
	x-r|r\right) f_{X^{\ast }|Z}\left( r|z\right) dr.
	\end{eqnarray*}%

	Also, notice that
	\begin{eqnarray*}
	q\left( x,z\right) &\equiv &E\left[ Y_{i}|X_{i}=x,Z_{i}=z\right] =E\left[ E%
	\left[ Y_{i}|X_{i}^{\ast },X_{i},Z_{i}\right] |X_{i}=x,Z_{i}=z\right] \\
	&=&E\left[ \rho (X_{i}^{\ast })|X_{i}=x,Z_{i}=z\right] ,
	\end{eqnarray*}%
	where the last equality follows from Assumptions~\ref{ass:NPID:exclusion}
	and~\ref{ass:NPID:Nondiffer ME}.

	We want to compute
	\begin{equation*}
	E\left[ \rho\left( X^{\ast }\right) |X=x,Z=z\right] =\frac{\int \rho\left(
	r\right) f_{X^{\ast }|XZ}\left( r|x,z\right) dr}{f_{X|Z}\left( x|z\right) }=%
	\frac{\int \rho\left( r\right) f_{\varepsilon |X^{\ast }}\left( x-r|r\right)
	f_{X^{\ast }|Z}\left( r|z\right) dr}{f_{X|Z}\left( x|z\right) }.
	\end{equation*}%
	Notice that on the RHS, the denominator is a special case of the numerator
	with $\rho\left( r\right) =1$ for all $r$. Thus, it is sufficient to focus on
	the numerator.

	Next, we consider $\int \eta \left( r\right) f_{\varepsilon |X^{\ast }}\left( x-r|r\right)dr$, where $\eta(\cdot)$ is a bounded function with $p$ bounded derivatives. Recall that we have $\varepsilon =\tau \xi$, so $f_{\varepsilon |X^{\ast }}\left(
	\varepsilon |r\right) =\frac{1}{\tau }f_{\xi |X^{\ast }}\left( \frac{%
	\varepsilon }{\tau }|r\right) $. Then,
	\begin{align*}
	\int \eta \left( r\right) f_{\varepsilon |X^{\ast }}\left( x-r|r\right)
	dr&=\int \eta \left( x-\tau u\right) \tau f_{\varepsilon |X^{\ast }}\left(
	\tau u|x-\tau u\right) du \\
	&=\int \eta \left( x-\tau u\right) f_{\xi |X^{\ast }}\left( u|x-\tau
	u\right) du,
	\end{align*}
	where we used the substitution $r=x-\tau u$. Using Assumption \ref{ass:NPID:smoothness}, we have (for all $u$)
	\begin{align*}
		\eta \left( x-\tau u\right) f_{\xi |X^{\ast }}\left( u|x-\tau u\right) = &\eta(x) f_{\xi|X^*}(u|x) + \sum_{\ell = 1}^{p-1} \frac{(-\tau)^\ell}{\ell!} u^{\ell} \nabla_x^{\ell} \{\eta(x) f_{\xi|X^*} (u|x)  \} \\ &+ \frac{(-\tau)^p}{p!} u^p \nabla_x^p \left\{\eta( x) f_{\xi|X^*}(u| x)\right\}\big\vert_{x = x - \tilde \tau (u) u}
	\end{align*}
	for some $\tilde \tau (u) \in (0, \tau)$. By boundedness of $\eta(\cdot)$ (and its derivatives) and Assumption \ref{ass:NPID:dominance}, all the terms on the RHS of the equation above are integrable, and consequently we have
	\begin{eqnarray*}
	&&\int \eta \left( x-\tau u\right) f_{\xi |X^{\ast }}\left( u|x-\tau
	u\right) du \\
	&=&\int \eta \left( x\right) f_{\xi |X^{\ast }}\left( u|x\right) du-\tau
	\int u\left\{ \eta ^{\prime }\left( x\right) f_{\xi |X^{\ast }}\left(
	u|x\right) +\eta \left( x\right) \nabla _{x}f_{\xi |X^{\ast }}\left(
	u|x\right) \right\} du \\
	&&+\frac{\tau ^{2}}{2}\int u^{2}\left\{ \eta ^{\prime \prime }\left(
	x\right) f_{\xi |X^{\ast }}\left( u|x\right) +2\eta ^{\prime }\left(
	x\right) \nabla _{x}f_{\xi |X^{\ast }}\left( u|x\right) +\eta \left(
	x\right) \nabla _{x}^2f_{\xi |X^{\ast }}\left( u|x\right) \right\}
	du + R_\eta (\tau) \\
	&=&\eta \left( x\right) -\tau \left\{ \eta ^{\prime }\left( x\right) E\left[
	\xi |X^{\ast }=x\right] +\eta \left( x\right) \nabla _{x}E\left[ \xi
	|X^{\ast }=x\right] \right\} \\
	&&+\frac{\tau ^{2}}{2}\left\{ \eta ^{\prime \prime }\left( x\right) E\left[
	\xi ^{2}|X^{\ast }=x\right] +2\eta ^{\prime }\left( x\right) \nabla _{x}E%
	\left[ \xi ^{2}|X^{\ast }=x\right] +\eta \left( x\right) \nabla _{x}^2E\left[
	\xi ^{2}|X^{\ast }=x\right] \right\} + R_\eta (\tau) \\
	&=&\eta \left( x\right) + \frac{1}{2}\{\eta ^{\prime \prime }\left( x\right) v\left( x\right) +2\eta ^{\prime }\left( x\right) \nabla
	_{x} v\left( x\right) +\eta \left( x\right) \nabla^2
	_{x} v\left( x\right)\} + R_\eta (\tau),
	\end{eqnarray*}
	where (for $p=4$)
	\begin{align*}
		R_\eta (\tau) = -\frac{\tau^3}{6} \int u^3 \nabla_x^3 \{\eta(x) f_{\xi|X^*}(u|x)\} du + \frac{\tau^4}{24} \int u^4 \nabla_x^4 \left\{\eta(\tilde x) f_{\xi|X^*}(u|\tilde x)\right\}\big\vert_{x = x - \tilde \tau (u) u} du.
	\end{align*}
	In the derivation above, the second equality uses $\int u^\ell \nabla_x^k f_{\xi|X^*}(u|x) du = \nabla_x^k \int u^\ell f_{\xi|X^*}(u|x) du = \nabla_x^k E [\xi^k|X^*=x]$ (differentiation under the integral sign is possible due to Assumptions \ref{ass:NPID:smoothness} and \ref{ass:NPID:dominance} for non-negative integers $k, \ell \leqslant p$), and the last equality is due to $E[\xi|X^*=x] = 0$ (as a function of $x$) and uses the notation $v(x) = E[\varepsilon^2|X^*=x] = \tau^2 E[\xi^2|X^*=x]$.

	Next, we establish $R_\eta (\tau) = O (\tau^p)$. Notice that
	\begin{eqnarray*}
		&&\int u^3 \nabla_x^3 \{\eta(x) f_{\xi|X^*}(u|x)\} du \\ &= &\int u^3 \{\eta'''(x) f_{\xi|X^*}(u|x) + 3 \eta''(x) \nabla_x f_{\xi|X^*}(u|x) + 3 \eta'(x) \nabla_{x}^2 f_{\xi|X^*} (u|x) + \eta(x) \nabla_{x}^3 f_{\xi|X^*} (u|x) \} du \\
		&=&\eta'''(x) E \left[\xi^3|X^*=x\right] + 3 \eta''(x) \nabla_x E \left[\xi^3|X^*=x\right] + 3 \eta'(x) \nabla_x^2 E \left[\xi^3|X^*=x\right] + \eta(x) \nabla_x^3 E \left[\xi^3|X^*=x\right] \\
		&=& 0,
	\end{eqnarray*}
	where the last equality uses $E[\xi^3|X^*=x] = 0$ (as a function of $x$). Finally,
	\begin{eqnarray*}
		&&\int u^4 \nabla_x^4 \left\{\eta(\tilde x) f_{\xi|X^*}(u|\tilde x)\right\}\big\vert_{\tilde x = x - \tilde \tau (u) u} du\\
		&=&\int u^4 \Big\{\eta''''(x - \tilde \tau (u) u) f_{\xi|X^*}(u|x - \tilde \tau (u) u) + 4 \eta'''(x - \tilde \tau (u) u) \nabla_x f_{\xi|X^*}(u|x - \tilde \tau (u) u) \\&& + 6 \eta''(x - \tilde \tau (u) u) \nabla_{x}^2 f_{\xi|X^*} (u|x - \tilde \tau (u) u) + 4 \eta'(x - \tilde \tau (u) u) \nabla_{x}^3 f_{\xi|X^*} (u|x - \tilde \tau (u) u) \\&& + \eta(x - \tilde \tau (u) u) \nabla_{x}^4 f_{\xi|X^*} (u|x - \tilde \tau (u) u) \Big\} du.
	\end{eqnarray*}	
	Since $\eta$ and its derivatives are uniformly bounded, for some $C > 0$, we have
	\begin{align*}
		\vert R_\eta (\tau) \vert \leqslant C \tau^4 \int u^4 \sup_{\tilde x \in \suppX} \left\{ f_{\xi|X^*} (u| \tilde x) + \sum_{\ell = 1}^4 \left \vert \nabla_x^\ell f_{\xi|X^*} (u| \tilde x) \right \vert  \right\} du = O(\tau^4),
	\end{align*}
	where the last equality uses Assumption \ref{ass:NPID:dominance}. Hence, we conclude
	\begin{align*}
		\int \eta \left( r\right) f_{\varepsilon |X^{\ast }}\left( x-r|r\right)
		dr = \eta \left( x\right) + \frac{1}{2}\{\eta ^{\prime \prime }\left( x\right) v\left( x\right) +2\eta ^{\prime }\left( x\right) \nabla
	_{x} v\left( x\right) +\eta \left( x\right) \nabla^2
	_{x} v\left( x\right)\} + O(\tau^p).
	\end{align*}
	Note that this result also implies
	\begin{align}
		\label{eq:pf:eta bias}
		\int \eta \left( r\right) f_{\varepsilon |X^{\ast }}\left( x-r|r\right) dr = \eta(x) + O(\tau^2).
	\end{align}

	Next, we apply the derived result to
	\begin{align*}
		q\left( x,z\right)  \notag = \frac{\int \rho \left( r\right) f_{X^{\ast }|Z}\left( r|z\right)
	f_{\varepsilon |X^{\ast }}\left( x-r|r\right) dr}{\int f_{X^{\ast }|Z}\left(
	r|z\right) f_{\varepsilon |X^{\ast }}\left( x-r|r\right) dr}.
	\end{align*}
	For the numerator, we use $\eta \left( x\right) =\rho \left( x\right)
	f_{X^{\ast }|Z}\left( x|z\right) $, $\eta ^{\prime }\left( x\right) =\rho
	\left( x\right) ^{\prime }f_{X^{\ast }|Z}\left( x|z\right) +\rho \left(
	x\right) f_{X^{\ast }|Z}^{\prime }\left( x|z\right) $, and $\eta ^{\prime
	\prime }\left( x\right) =\rho ^{\prime \prime }\left( x\right) f_{X^{\ast
	}|Z}\left( x|z\right) +2\rho ^{\prime }\left( x\right) f_{X^{\ast
	}|Z}^{\prime }\left( x|z\right) +\rho \left( x\right) f_{X^{\ast
	}|Z}^{\prime \prime }\left( x|z\right) $. Thus,%
	\begin{align*}
		&\int \rho \left( r\right) f_{X^{\ast }|Z}\left( r|z\right) f_{\varepsilon |X^{\ast }}\left( x-r|r\right) dr \\ =  &\rho(x) \left(f_{X^*|Z} (x|z) + \frac{1}{2}\{f_{X^*|Z} ^{\prime \prime }\left( x|z\right) v\left( x\right) +2f_{X^*|Z} ^{\prime }\left( x|z\right) \nabla
	_{x} v\left( x\right) +f_{X^*|Z} \left( x|z\right) \nabla^2
	_{x} v\left( x\right)\}\right) \\
	& +
	\frac{1}{2}\left( \rho ^{\prime \prime }\left( x\right) f_{X^{\ast
	}|Z}\left( x|z\right) +2\rho ^{\prime }\left( x\right) f_{X^{\ast
	}|Z}^{\prime }\left( x|z\right) \right) v(x) +\rho \left( x\right) ^{\prime
	}f_{X^{\ast }|Z}\left( x|z\right) \nabla _{x}v\left( x\right) + O (\tau^p). 
	\end{align*}
	Similarly,
	\begin{align*}
		f_{X|Z} (x|z) &=\int f_{X^{\ast }|Z}\left(r|z\right) f_{\varepsilon |X^{\ast }}\left( x-r|r\right) dr \\ &=  f_{X^*|Z} (x|z) + \frac{1}{2}\{f_{X^*|Z} ^{\prime \prime }\left( x|z\right) v\left( x\right) +2f_{X^*|Z} ^{\prime }\left( x|z\right) \nabla
	_{x} v\left( x\right) +f_{X^*|Z} \left( x|z\right) \nabla^2
	_{x} v\left( x\right)\} + O(\tau^p).
	\end{align*}
	Hence, using $\nabla_x^{\ell} v(x) = O (\tau^2)$ for $l \in \{0,1,2\}$, we conclude
	\begin{eqnarray*}
	q\left( x,z\right) &=&\rho \left( x\right) + \frac{\frac{1}{2}\left( \rho ^{\prime \prime }\left( x\right) f_{X^{\ast
	}|Z}\left( x|z\right) +2\rho ^{\prime }\left( x\right) f_{X^{\ast
	}|Z}^{\prime }\left( x|z\right) \right) v\left( x\right) + \rho \left( x\right) ^{\prime
	}f_{X^{\ast }|Z}\left( x|z\right) \nabla _{x}v\left( x\right)  }{%
	f_{X^{\ast }|Z}\left( x|z\right) + O(\tau^2) }+O\left( \tau ^{p}\right)  \notag \\
	&=&\rho \left( x\right)  + \left(\frac{1}{2}\rho ^{\prime \prime }\left( x\right) +\rho ^{\prime
	}\left( x\right) s_{X^{\ast
	}|Z}\left( x|z\right)\right) v(x) +\rho ^{\prime }\left( x\right) v'(x) +O\left( \tau ^{p}\right) .
	\end{eqnarray*}
\end{proof}

\begin{proof}[Proof of Lemma \ref{lem:q(x,z) expansion}, Part (\ref{item: s q prime bias})]
	First, we consider
	\begin{align*}
		\nabla_x \int \eta(r) f_{\varepsilon|X^*} (x - r|r) dr  &= \nabla_x \int \eta (x - \tau u) f_{\xi|X^*} (u|x - \tau u) du\\
		&= \int \nabla_x \{\eta (x - \tau u) f_{\xi|X^*} (u|x - \tau u)\} du,
	\end{align*}
	where the second equality follows from Assumptions \ref{ass:NPID:smoothness} and \ref{ass:NPID:dominance}. Analogously to the expansion of $\eta(x-\tau u) f_{\xi|X^*}(u|x - \tau u)$ considered in the proof of Part \eqref{item: q(x,z) expansion}, we have
	\begin{align*}
		 \nabla_x \{\eta (x - \tau u) f_{\xi|X^*} (u|x - \tau u)\} = &\eta'(x) f_{\xi|X^*} (u|x) + \eta(x) \nabla_x f_{\xi|X^*} (u|x) \\ &- \tau u \nabla_x^2 \{\eta (x) f_{\xi|X^*} (u|x)\} \\ &+ \frac{\tau^2}{2} u^2 \nabla_x^3 \{\eta (x - \tilde \tau (u) u) f_{\xi|X^*} (u|x - \tilde \tau(u) u)\},
	\end{align*}
	for some $\tilde \tau (u) \in (0,\tau)$. %
	Hence,
	\begin{align*}
		\nabla_x \int \eta(r) f_{\varepsilon|X^*} (x - r|r) dr  &= \int \left(\eta'(x) f_{\xi|X^*} (u|x) + \eta(x) \nabla_x f_{\xi|X^*} (u|x)\right) du + O (\tau^2) \\ &= \eta'(x) + O(\tau^2).
	\end{align*}

	Plugging $\rho (x) f_{X^*} (x)$, $f_{X^*}(x)$, and  $f_{X^*|Z} (x|z)$ as $\eta(x)$, we obtain
	\begin{alignat*}{3}
			&&&\nabla_x \int \rho(r) f_{X^*}(r) f_{\varepsilon|X^*} (x - r|r) dr &&= \nabla_x \{\rho (x) f_{X^*} (x) \} + O(\tau^2), \\
		f_X'(x) &= &&\nabla_x \int f_{X^*}(r) f_{\varepsilon|X^*} (x - r|r) dr &&= f_{X^*}' (x) + O(\tau^2), \\
		f_{X|Z}'(x,z) &= &&\nabla_x \int f_{X^*|Z}(r|z) f_{\varepsilon|X^*} (x - r|r) dr &&= f_{X^*|Z}' (x|z) + O(\tau^2).
	\end{alignat*}
	Similarly, using equation \eqref{eq:pf:eta bias} derived in the proof of Part \eqref{item: q(x,z) expansion},
	\begin{alignat*}{3}
		&&&\int \rho(r) f_{X^*}(r) f_{\varepsilon|X^*} (x - r|r) dr &&=  \rho (x) f_{X^*} (x) + O(\tau^2), \\
		f_X(x) &= && \int f_{X^*}(r) f_{\varepsilon|X^*} (x - r|r) dr &&= f_{X^*} (x) + O(\tau^2), \\
		f_{X|Z}(x,z) &= && \int f_{X^*|Z}(r|z) f_{\varepsilon|X^*} (x - r|r) dr &&= f_{X^*|Z} (x|z) + O(\tau^2).
	\end{alignat*}

	Using the results above, we establish
	\begin{align}
		s_{X|Z}(x|z) - s_{X^*|Z}(x|z) &= \frac{f'_{X|Z}(x|z) f_{X^*|Z}(x|z) - f'_{X|Z}(x|z) f_{X^*|Z}(x|z)}{f_{X|Z}(x|z) f_{X^*|Z}(x|z)} \notag \\
		&= \frac{O (\tau^2)}{f_{X^*|Z}(x|z)^2 + O (\tau^2)} = O (\tau^2), \label{eq:pf:score bias}
	\end{align}
	Similarly,
	\begin{align}
		q'(x) &= \nabla_x \left\{ \frac{\int \rho(r) f_{X^*}(r) f_{\varepsilon|X^*} (x - r|r) dr}{\int f_{X^*}(r) f_{\varepsilon|X^*} (x - r|r) dr} \right\} \notag \\
		&= \frac{\left(\rho'(x) f_{X^*}(x) + \rho(x) f_{X^*}'(x) + O (\tau^2)\right) \left(f_{X^*}(x) + O(\tau^2)\right) - \left(f'_{X^*}(x) + O(\tau^2)\right) \left(\rho(x) f_{X^*}(x) + O(\tau^2)  \right)  }{\left(f_{X^*}(x) + O(\tau^2)\right)^2} \notag \\
		&= \rho'(x) + O (\tau^2). \label{eq:pf:q prime bias}
	\end{align}
	which completes the proof.
\end{proof}

\bigskip

Before proving additional auxiliary results we introduce an additional notation.

\begin{definition*}
	Consider a function $g: \mathcal X \times \mathbb R^+ \rightarrow \mathbb R$ such that $\nabla_x g(x, \tau)$ exists (for every $\tau \in \mathbb R^+$). We say that $g (x, \tau) = O_x (\tau^p)$ if $g(x, \tau) = O (\tau^p)$ and $\nabla_x g(x, \tau) = O(\tau^p)$ as $\tau \downarrow 0$.
\end{definition*}

\begin{lemma}
	\label{lem:smooth q prime and s prime}
	Suppose that the hypotheses of Theorem \ref{thm:NPID-non-cl:IV} are satisfied. Then, we have $f_{X|Z}(x,z) = f_{X^*|Z} (x|z) + O_x (\tau^2)$, $s_{X
		|Z}(x|z) = s_{X^*|Z}(x|z) + O_x (\tau^{p-2})$, and $q'(x) = \rho'(x) + O_x (\tau^{p-2})$.
\end{lemma}

\begin{proof}[Proof of Lemma \ref{lem:smooth q prime and s prime}]
	For concreteness, we provide the proof for $p=4$. The proof for $p=3$ is analogous (but simpler).

	First, consider $\int \eta(r) f_{\varepsilon|X^*} (x - r|r) dr = \int \eta (x - \tau u) f_{\xi|X^*} (u|x - \tau u) du$, where $\eta(\cdot)$ is a bounded function with $p$ bounded derivatives. Considering an expansion analogous to the one provided in the proof of Lemma \ref{lem:q(x,z) expansion}, Part \eqref{item: q(x,z) expansion}, we obtain
	\begin{align*}
		\eta (x - \tau u) f_{\xi|X^*} (u|x - \tau u) = &\eta(x) f_{\xi|X^*} (u|x) - \tau u \nabla_x \{\eta(x) f_{\xi|X^*} (u|x)\} \\ &+ \int_{0}^\tau u^2 \nabla_x^2 \{\eta(x - tu) f_{\xi|X^*}(u|x - tu)\} (\tau - t) dt,
	\end{align*}
	where the remainder is given in the integral form. Then,
	\begin{align*}
		\int \eta(r) f_{\varepsilon|X^*} (x - r|r) dr &= \int \eta (x - \tau u) f_{\xi|X^*} (u|x - \tau u) du \\
		&= \eta(x) + R_{\eta} (x;\tau),
	\end{align*}
	where
	\begin{align*}
		R_{\eta} (x;\tau) = \int \left(\int_{0}^\tau u^2 \nabla_x^2 \{\eta(x - tu) f_{\varepsilon|X^*}(u|x - tu)\} (\tau - t) dt\right) du.%
	\end{align*}
	Since $\eta$ and its derivatives are bounded, for some $C > 0$, we have
	\begin{align*}
		\abs{\int_{0}^\tau u^2 \nabla_x^2 \{\eta(x - tu) f_{\varepsilon|X^*}(u|x - tu)\} (\tau - t) dt} \leqslant \frac{\tau^2}{2} \abs{u}^2 \sup_{\tilde x \in \suppX} \sum_{\ell = 0}^2 \abs{\nabla_x^{\ell} f_{\xi|X^*} (u|\tilde x) }.
	\end{align*}
	Combining this result with Assumption \ref{ass:NPID:dominance}, we obtain $R_{\eta} (x;\tau) = O (\tau^2)$. Similarly, we conclude
	\begin{align*}
		\nabla_x R_{\eta} (x;\tau) = \int \left(\int_{0}^\tau u^2 \nabla_x^3 \{\eta(x - tu) f_{\varepsilon|X^*}(u|x - tu)\} (\tau - t) dt\right) du = O (\tau^2).
	\end{align*}
	Hence, we conclude
	\begin{align}
		\label{eq:pf:smooth eta ord 2}
		\int \eta(r) f_{\varepsilon|X^*} (x - r|r) dr = \eta(x) + O_x (\tau^2).
	\end{align}

	Next, consider $\nabla_x \int \eta(r) f_{\varepsilon|X^*} (x - r|r) dr = \nabla_x \int \eta (x - \tau u) f_{\xi|X^*} (u|x - \tau u) du $, where $\eta(\cdot)$ is a bounded function with $p$ bounded derivatives. Considering an expansion analogous to the one derived in the proof of Lemma \ref{lem:q(x,z) expansion}, Part \eqref{item: s q prime bias}, we obtain
	\begin{align*}
		\nabla_x \int \eta(r) f_{\varepsilon|X^*} (x - r|r) dr  &= \nabla_x \int \eta (x - \tau u) f_{\xi|X^*} (u|x - \tau u) du \\ &= \eta'(x) + r_\eta (x;\tau),
	\end{align*}
	where
	\begin{align*}
		r_{\eta} (x, \tau) &= \int \left(\int_{0}^\tau u^2 \nabla_x^3 \{\eta(x - tu) f_{\xi|X^*}(u|x - tu) \} (\tau - t) dt\right) du = \nabla_x R_{\eta} (x;\tau) = O (\tau^2).
	\end{align*}
	Similarly,
	\begin{align*}
		\nabla_x r_\eta (x, \tau) = \int \left(\int_{0}^\tau u^2 \nabla_x^4 \{\eta(x - tu) f_{\xi|X^*}(u|x - tu) \} (\tau - t) dt\right) du = O (\tau^2),
	\end{align*}
	which demonstrates
	\begin{align}
		\label{eq:pf:smooth nabla eta ord 2}
		\nabla_x \int \eta(r) f_{\varepsilon|X^*} (x - r|r) dr &= \eta'(x) + O_x (\tau^2).
	\end{align}

	Applying \eqref{eq:pf:smooth eta ord 2} and \eqref{eq:pf:smooth nabla eta ord 2} with $\rho (x) f_{X^*} (x)$, $f_{X^*}(x)$, and  $f_{X^*|Z} (x|z)$ as $\eta(x)$ we obtain
	\begin{alignat*}{3}
			&&&\int \rho(r) f_{X^*}(r) f_{\varepsilon|X^*} (x - r|r) dr &&= \rho (x) f_{X^*} (x) + O_x (\tau^2), \\
		f_X(x) &= &&\int f_{X^*}(r) f_{\varepsilon|X^*} (x - r|r) dr &&= f_{X^*} (x) + O_x (\tau^2), \\
		f_{X|Z}(x,z) &= &&\int f_{X^*|Z}(r|z) f_{\varepsilon|X^*} (x - r|r) dr &&= f_{X^*|Z} (x|z) + O_x (\tau^2),
	\end{alignat*}
	and
	\begin{alignat*}{3}
			&&&\nabla_x \int \rho(r) f_{X^*}(r) f_{\varepsilon|X^*} (x - r|r) dr &&= \nabla_x \{\rho (x) f_{X^*} (x) \} + O_x (\tau^2), \\
		f_X'(x) &= &&\nabla_x \int f_{X^*}(r) f_{\varepsilon|X^*} (x - r|r) dr &&= f_{X^*}' (x) + O_x (\tau^2), \\
		f_{X|Z}'(x,z) &= &&\nabla_x \int f_{X^*|Z}(r|z) f_{\varepsilon|X^*} (x - r|r) dr &&= f_{X^*|Z}' (x|z) + O_x (\tau^2).
	\end{alignat*}

	Combining the results above with derivations as in equations \eqref{eq:pf:score bias} and \eqref{eq:pf:q prime bias} completes the proof.
\end{proof}

\begin{lemma}
\label{lem:diff q} Suppose that the hypotheses of Theorem \ref{thm:NPID-non-cl:IV} are satisfied. Then,
\begin{equation*}
q\left( x,z_{1}\right) -q\left( x,z_{2}\right) = \rho ^{\prime
}\left( x\right) v\left( x\right) \left( s_{X^{\ast }|Z}^{\prime }\left(
x|z_{1}\right) -s_{X^{\ast }|Z}^{\prime }\left( x|z_{2}\right) \right)
+O_x\left( \tau ^{p}\right) .
\end{equation*}
\end{lemma}

\begin{proof}[Proof of Lemma~\protect\ref{lem:diff q}]
For concreteness, we provide the proof for $p=4$. The proof for $p=3$ is analogous (but simpler).

\noindent \textbf{1. }Note that by bringing the difference of ratios to the
common denominator we can write%
\begin{equation*}
q\left( x,z_{1}\right) -q\left( x,z_{2}\right) =\frac{N\left(
x,z_{1},z_{2},\tau \right) }{D\left( x,z_{1},z_{2},\tau \right) },
\end{equation*}%
where $D\left( x,z_{1},z_{2},\tau \right) =f_{X|Z}\left( x|z_{1}\right)
f_{X|Z}\left( x|z_{2}\right) $. The numerator is%
\begin{eqnarray*}
&&N\left( x,z_{1},z_{2},\tau \right) \\
&\equiv &\int \rho \left( r\right) f_{X^{\ast }|Z}\left( r|z_{1}\right)
f_{\varepsilon |X^{\ast }}\left( x-r|r\right) dr\times \int f_{X^{\ast
}|Z}\left( r|z_{2}\right) f_{\varepsilon |X^{\ast }}\left( x-r|r\right) dr \\
&&\qquad -\int \rho \left( r\right) f_{X^{\ast }|Z}\left( r|z_{2}\right)
f_{\varepsilon |X^{\ast }}\left( x-r|r\right) dr\times \int f_{X^{\ast
}|Z}\left( r|z_{1}\right) f_{\varepsilon |X^{\ast }}\left( x-r|r\right) dr \\
&=&\iint \int \left\{ \rho \left( x-\tau u_{1}\right) -\rho \left( x-\tau
u_{2}\right) \right\} f_{X^{\ast }|Z}\left( x-\tau u_{1}|z_{1}\right)
f_{X^{\ast }|Z}\left( x-\tau u_{2}|z_{2}\right) \\
&&\qquad \times f_{\xi |X^{\ast }}\left( u_{1}|x-\tau u_{1}\right) f_{\xi
|X^{\ast }}\left( u_{2}|x-\tau u_{2}\right) du_{1}du_{2} \\
&=&\iint \left\{ \rho \left( x-\tau u_{1}\right) -\rho \left( x-\tau
u_{2}\right) \right\} f_{\xi ,X^{\ast }|Z}\left( u_{1},x-\tau
u_{1}|z_{1}\right) f_{\xi, X^{\ast }|Z}\left( u_{2},x-\tau u_{2}|z_{2}\right)
du_{1}du_{2}.
\end{eqnarray*}

Note that for any $q_{1},q_{2},Q_{1},Q_{2}$ we have%
\begin{equation*}
q_{1}q_{2}=\left( q_{1}-Q_{1}\right) Q_{2}+Q_{1}\left( q_{2}-Q_{2}\right)
+\left( q_{1}-Q_{1}\right) \left( q_{2}-Q_{2}\right) +Q_{1}Q_{2},
\end{equation*}%
and taking $q_{j}\equiv f_{\xi ,X^{\ast }|Z}\left( u_{j},x-\tau
u_{j}|z_{j}\right) $, $Q_{j}=f_{\xi ,X^{\ast }|Z}\left( u_{j},x|z_{j}\right)
-f_{\xi ,X^{\ast }|Z}^{\prime }\left( u_{j},x|z_{j}\right) \tau u_{j}$ for $%
j\in \left[ 2\right] $ we have%
\begin{eqnarray*}
&&f_{\xi ,X^{\ast }|Z}\left( u_{1},x-\tau u_{1}|z_{1}\right) f_{\xi X^{\ast
}|Z}\left( u_{2},x-\tau u_{2}|z_{2}\right) \\
&=&T_{ff,1}+T_{ff,2}+T_{ff,3}+T_{ff,4} \\
T_{ff,1} &\equiv &\left( f_{\xi ,X^{\ast }|Z}\left( u_{1},x-\tau
u_{1}|z_{1}\right) +f_{\xi ,X^{\ast }|Z}^{\prime }\left(
u_{1},x|z_{1}\right) \tau u_{1}-f_{\xi ,X^{\ast }|Z}\left(
u_{1},x|z_{1}\right) \right) \\
&&\qquad \times \left( -f_{\xi ,X^{\ast }|Z}^{\prime }\left(
u_{2},x|z_{2}\right) \tau u_{2}+f_{\xi ,X^{\ast }|Z}\left(
u_{2},x|z_{2}\right) \right) \\
T_{ff,2} &\equiv &\left( f_{\xi ,X^{\ast }|Z}\left( u_{2},x-\tau
u_{2}|z_{2}\right) +f_{\xi ,X^{\ast }|Z}^{\prime }\left(
u_{2},x|z_{2}\right) \tau u_{2}-f_{\xi ,X^{\ast }|Z}\left(
u_{2},x|z_{2}\right) \right) \\
&&\qquad \times \left( -f_{\xi ,X^{\ast }|Z}^{\prime }\left(
u_{1},x|z_{1}\right) \tau u_{1}+f_{\xi ,X^{\ast }|Z}\left(
u_{1},x|z_{1}\right) \right) \\
T_{ff,3} &\equiv &\left( f_{\xi ,X^{\ast }|Z}\left( u_{1},x-\tau
u_{1}|z_{1}\right) +f_{\xi ,X^{\ast }|Z}^{\prime }\left(
u_{1},x|z_{1}\right) \tau u_{1}-f_{\xi ,X^{\ast }|Z}\left(
u_{1},x|z_{1}\right) \right) \\
&&\qquad \times \left( f_{\xi X^{\ast }|Z}\left( u_{2},x-\tau
u_{2}|z_{2}\right) +f_{\xi ,X^{\ast }|Z}^{\prime }\left(
u_{2},x|z_{2}\right) \tau u_{2}-f_{\xi ,X^{\ast }|Z}\left(
u_{2},x|z_{2}\right) \right) \\
T_{ff,4} &\equiv &\left( -f_{\xi ,X^{\ast }|Z}^{\prime }\left(
u_{1},x|z_{1}\right) \tau u_{1}+f_{\xi ,X^{\ast }|Z}\left(
u_{1},x|z_{1}\right) \right) \left( -f_{\xi ,X^{\ast }|Z}^{\prime }\left(
u_{2},x|z_{2}\right) \tau u_{2}+f_{\xi ,X^{\ast }|Z}\left(
u_{2},x|z_{2}\right) \right)
\end{eqnarray*}%
and let%
\begin{equation*}
I_{ff,j}\left( x\right) \equiv \iint \left\{ \rho \left( x-\tau u_{1}\right)
-\rho \left( x-\tau u_{2}\right) \right\} T_{ff,j}\left( \ldots \right)
du_{1}du_{2},\quad j\in \left[ 4\right] .
\end{equation*}

\noindent \textbf{2.} Note that 
\begin{subequations}
\label{eq:pf:lem:diff q: f approx}
\begin{eqnarray}
&&f_{\xi ,X^{\ast }|Z}\left( u_{j},x-\tau u_{j}|z_{j}\right) +f_{\xi
,X^{\ast }|Z}^{\prime }\left( u_{j},x|z_{j}\right) \tau u_{j}-f_{\xi
,X^{\ast }|Z}\left( u_{j},x|z_{j}\right)  \notag \\
&=&\int_{0}^{\tau } f_{\xi ,X^{\ast }|Z}^{\prime \prime }\left(
u_{j},x-tu_{j}|z_{j}\right) u_{j}^{2}\left( \tau -t\right) dt
\label{eq:pf:lem:diff q: f approx ord 2} \\
&=&\frac{\tau ^{2}}{2}f_{\xi ,X^{\ast }|Z}^{\prime \prime }\left(
u_{j},x|z_{j}\right) u_{j}^{2} - \int_{0}^{\tau }\frac{1}{2}f_{\xi ,X^{\ast
}|Z}^{\prime \prime \prime }\left( u_{j},x-tu_{j}|z_{j}\right)
u_{j}^{3}\left( \tau -t\right) ^{2}dt
\label{eq:pf:lem:diff q: f approx ord 3}
\end{eqnarray}%
and 
\end{subequations}
\begin{subequations}
\label{eq:pf:lem:diff q: d rho approx}
\begin{eqnarray}
&&\rho \left( x-\tau u_{1}\right) -\rho \left( x-\tau u_{2}\right)  \notag \\
&=& \sum_{j=1}^2 \left(-1\right)^j \int_{0}^{\tau} \rho'(x - t u_j ) u_j dt  \label{eq:pf:lem:diff q: d rho approx ord 1} \\
&=&\tau \rho ^{\prime }\left( x\right) \left( u_{2}-u_{1}\right)
+\sum_{j=1}^{2}\left( -1\right) ^{j-1}\int_{0}^{\tau } \rho
^{\prime \prime }\left( x-tu_{j}\right) u_{j}^{2}\left( \tau -t\right)
dt  \label{eq:pf:lem:diff q: d rho approx ord 2} \\
&=&\tau \rho ^{\prime }\left( x\right) \left( u_{2}-u_{1}\right) +\frac{\tau
^{2}}{2}\rho ^{\prime \prime }\left( x\right) \left(
u_{1}^{2}-u_{2}^{2}\right) \notag \\&& +\sum_{j=1}^{2}\left( -1\right)
^{j}\int_{0}^{\tau }\frac{1}{2}\rho ^{\prime \prime \prime }\left(
x-tu_{j}\right) u_{j}^{3}\left( \tau -t\right) ^{2}dt.
\label{eq:pf:lem:diff q: d rho approx ord 3}
\end{eqnarray}

We will use the notation 
\end{subequations}
\begin{eqnarray*}
\mathbb{M}_{|\zeta _{1}\zeta _{2}z_{1}z_{2}}\left[ a\left(
x,U_{1},U_{2}\right) \right] &\equiv &\mathbb{E}_{|\zeta _{1}\zeta
_{2}z_{1}z_{2}}\left[ a\left( x,U_{1},U_{2}\right) \right] f_{X^{\ast
}|Z}\left( \zeta _{1}|z_{1}\right) f_{X^{\ast }|Z}\left( \zeta
_{2}|z_{2}\right) , \\
\mathbb{E}_{|\zeta _{1}\zeta _{2}z_{1}z_{2}}\left[ a\left(
x,U_{1},U_{2}\right) \right] &\equiv &\iint a\left( x,u_{1},u_{2}\right)
f_{\xi |X^{\ast },Z}\left( u_{1}|\zeta _{1},z_{1}\right) f_{\xi |X^{\ast
},Z}\left( u_{2}|\zeta _{2},z_{2}\right) du_{1}du_{2},
\end{eqnarray*}%
i.e., $\mathbb{E}_{|\zeta _{1}\zeta _{2}z_{1}z_{2}}\left[ \cdot \right] $ is
the expectation w.r.t. $U_{1}$ and $U_{2}$, where $U_{1}$ and $U_{2}$ are
independent and $U_{j}\sim f_{\xi |X^{\ast },Z}\left( \cdot |\zeta
_{j},z_{j}\right) = f_{\xi|X^*}(\cdot|\zeta_j) $ for $j\in \left[ 2\right] $.

Let 
\begin{align*}
\Delta \left( x,q,s\right) &\equiv \left\{ \rho \left( x-\tau u_{1}\right)
-\rho \left( x-\tau u_{2}\right) \right\} \times \nabla _{x}^{q}f_{\xi
,X^{\ast }|Z}\left( u_{1},x|z_{1}\right) u_{1}^{q}\times \nabla
_{x}^{s}f_{\xi ,X^{\ast }|Z}\left( u_{2},x|z_{2}\right) u_{2}^{s}, \\
I_{\Delta \left( x,q,s\right) } &\equiv \iint \Delta \left( x,q,s\right)
du_{1}du_{2}=\left. \nabla _{\zeta _{1}}^{q}\nabla _{\zeta _{2}}^{s}\mathbb{M%
}_{|\zeta _{1}\zeta _{2}z_{1}z_{2}}\left[ \left( \rho \left( x-\tau
U_{1}\right) -\rho \left( x-\tau U_{2}\right) \right) U_{1}^{q}U_{2}^{s}%
\right] \right\vert _{\zeta _{1}=\zeta _{2}=x}.
\end{align*}%
Also, let 
\begin{eqnarray*}
h\left( x,a,b,q,s\right) &\equiv &u_{1}^{a}u_{2}^{b}\times \nabla
_{x}^{q}f_{\xi ,X^{\ast }|Z}\left( u_{1},x|z_{1}\right) u_{1}^{q}\times
\nabla _{x}^{s}f_{\xi ,X^{\ast }|Z}\left( u_{2},x|z_{2}\right) u_{2}^{s} \\
&=&u_{1}^{a+q}u_{2}^{b+s}\times \nabla _{x}^{q}f_{\xi ,X^{\ast }|Z}\left(
u_{1},x|z_{1}\right) \nabla _{x}^{s}f_{\xi ,X^{\ast }|Z}\left(
u_{2},x|z_{2}\right)
\end{eqnarray*}%
and 
\begin{equation}
I_{h\left( x,a,b,q,s\right) }\equiv \iint h\left( x,a,b,q,s\right)
du_{1}du_{2}=\left. \nabla _{\zeta _{1}}^{q}\nabla _{\zeta _{2}}^{s}\mathbb{M%
}_{|\zeta _{1}\zeta _{2}z_{1}z_{2}}\left[ U_{1}^{a+q}U_{2}^{b+s}\right]
\right\vert _{\zeta _{1}=\zeta _{2}=x}.  \label{eq:pf:lem:diff q: int h}
\end{equation}

From equation~(\ref{eq:pf:lem:diff q: d rho approx ord 3}) we have%
\begin{eqnarray*}
&&\Delta \left( x,q,s\right) \\
&=&\left\{ \tau \rho ^{\prime }\left( x\right) \left( u_{2}-u_{1}\right) +%
\frac{\tau ^{2}}{2}\rho ^{\prime \prime }\left( x\right) \left(
u_{1}^{2}-u_{2}^{2}\right) +\sum_{j=1}^{2}\left( -1\right)
^{j}\int_{0}^{\tau }\frac{1}{2!}\rho ^{\prime \prime \prime }\left(
x-tu_{j}\right) u_{j}^{3}\left( \tau -t\right) ^{2}dt\right\} \\ &&\qquad \times h\left(
x,0,0,q,s\right) \\
&=&\tau \rho ^{\prime }\left( x\right) \left\{ h\left( x,0,1,q,s\right)
-h\left( x,1,0,q,s\right) \right\} +\frac{\tau ^{2}}{2}\rho ^{\prime \prime
}\left( x\right) \left\{ h\left( x,2,0,q,s\right) -h\left( x,0,2,q,s\right)
\right\} \\
&&\qquad +\sum_{j=1}^{2}\left( -1\right) ^{j}\int_{0}^{\tau }\frac{1}{2}%
\rho ^{\prime \prime \prime }\left( x-tu_{j}\right) u_{j}^{3}\left( \tau
-t\right) ^{2}dt\times h\left( x,0,0,q,s\right).
\end{eqnarray*}%
Thus, for $q,s \in \{0,1\}$, we have
\begin{eqnarray}
I_{\Delta \left( x,q,s\right) } &=&\tau \rho ^{\prime }\left( x\right)
\left\{ I_{h\left( x,0,1,q,s\right) }-I_{h\left( x,1,0,q,s\right) }\right\} \notag \\ && \qquad +%
\frac{\tau ^{2}}{2}\rho ^{\prime \prime }\left( x\right) \left\{ I_{h\left(
x,2,0,q,s\right) }-I_{h\left( x,0,2,q,s\right) }\right\} +O_x\left(
 \tau ^{3}\right).
\label{eq:pf:lem:diff q: int psi approx ord 3}%
\end{eqnarray}

Notice that if $q = 2$ or $s = 2$, we use equation~(\ref{eq:pf:lem:diff q: d rho approx ord 2}) instead to obtain %
\begin{eqnarray*}
&&\Delta \left( x,q,s\right) \\
&=&\left\{ \tau \rho ^{\prime }\left( x\right) \left( u_{2}-u_{1}\right)
+\sum_{j=1}^{2}\left( -1\right) ^{j-1}\int_{0}^{\tau } \rho
^{\prime \prime }\left( x-tu_{j}\right) u_{j}^{2}\left( \tau -t\right)
dt\right\} h\left(
x,0,0,q,s\right) \\
&=&\tau \rho ^{\prime }\left( x\right) \left\{ h\left( x,0,1,q,s\right)
-h\left( x,1,0,q,s\right) \right\}\\
&&\qquad + \sum_{j=1}^{2}\left( -1\right) ^{j-1}\int_{0}^{\tau } \rho
^{\prime \prime }\left( x-tu_{j}\right) u_{j}^{2}\left( \tau -t\right)
dt \times h\left(
x,0,0,q,s\right).
\end{eqnarray*}%
Thus, for both $q,s \in \{0,1,2\}$, we have
\begin{eqnarray}
I_{\Delta \left( x,q,s\right) }
=\tau \rho ^{\prime }\left( x\right) \left\{ I_{h\left( x,0,1,q,s\right)
}-I_{h\left( x,1,0,q,s\right) }\right\} +O_x\left( 
\tau ^{2}\right).  \label{eq:pf:lem:diff q: int psi approx ord 2}
\end{eqnarray}

\bigskip

\noindent \textbf{3.1.} Using equation~(\ref{eq:pf:lem:diff q: f approx ord
2}) we have 

\begin{equation*}
T_{ff,3}\left( x\right) =\int_{0}^{\tau }f_{\xi ,X^{\ast
}|Z}^{\prime \prime }\left( u_{1},x-t_1 u_{1}|z_{1}\right) u_{1}^{2}\left( \tau
-t_1\right) dt_1\times \int_{0}^{\tau }f_{\xi ,X^{\ast }|Z}^{\prime
\prime }\left( u_{2},x-t_{2}u_{2}|z_{2}\right) u_{2}^{2}\left( \tau
-t_{2}\right) dt_{2}.
\end{equation*}%
First, 
\begin{align*}
	\abs{T_{ff,3}(x)} \leqslant u_1^2 \sup_{\tilde x \in \suppX} \abs{f_{\xi,X^*|Z}''(u_1, \tilde x|z_1)} u_2^2 \sup_{\tilde x \in \suppX} \abs{f_{\xi,X^*|Z}''(u_2, \tilde x|z_2)} \frac{\tau^4}{4}.
\end{align*}
Next,
\begin{align*}
	f_{\xi,X^*|Z}''(u, x|z) &= \nabla_x^{2} \{ f_{\xi|X^*} (u|x) f_{X^*|Z} (x|z) \} \\
	&= f_{\xi|X^*}'' (u|x) f_{X^*|Z} (x|z) + 2 f_{\xi|X^*}' (u|x) f_{X^*|Z}' (x|z) + f_{\xi|X^*} (u|x) f''_{X^*|Z} (x|z).
\end{align*}
Hence, using Assumption \ref{ass:NPID:dominance} and boundedness of $\rho$ and $f_{\xi|X^*}$ and its derivatives, we conclude
\begin{align*}
	I_{ff,3} (x) = O (\tau^4).
\end{align*}
Similarly,
\begin{align*}
	\abs{T_{ff,3}'(x)} \leqslant \sum_{j=0}^1 \sup_{\tilde x \in \suppX} \abs{\nabla_x^{3-j} f_{\xi,X^*|Z}(u_1, \tilde x|z_1)}  \sup_{\tilde x \in \suppX} \abs{\nabla_x^{2+j} f_{\xi,X^*|Z}(u_2, \tilde x|z_2)} u_1^2 u_2^2 \frac{\tau^4}{4},
\end{align*}
and
\begin{align*}
	I_{ff,3}'(x) = \iint \sum_{j=0}^1 \left\{ \nabla_x^{1-j} \rho \left( x-\tau u_{1}\right) - \nabla_x^{1-j} \rho \left( x-\tau u_{2}\right) \right\} \nabla_x^{j} T_{ff,3}\left( x \right) du_{1}du_{2} = O (\tau^4),
\end{align*}
where the second inequality follows from Assumption \ref{ass:NPID:dominance}, and boundedness of $\rho$ and $f_{X^*|Z}$ (and their derivatives). Hence we conclude $I_{ff,3} (x) = O_x(\tau^4)$.

\noindent \textbf{3.2.}%
\begin{align*}
T_{ff,1}\equiv &\left( f_{\xi ,X^{\ast }|Z}\left( u_{1},x-\tau
u_{1}|z_{1}\right) +f_{\xi ,X^{\ast }|Z}^{\prime }\left(
u_{1},x|z_{1}\right) \tau u_{1}-f_{\xi ,X^{\ast }|Z}\left(
u_{1},x|z_{1}\right) \right) \\ &\times \left( -f_{\xi ,X^{\ast }|Z}^{\prime }\left(
u_{2},x|z_{2}\right) \tau u_{2}+f_{\xi ,X^{\ast }|Z}\left(
u_{2},x|z_{2}\right) \right). 
\end{align*}%
Thus,%
\begin{eqnarray*}
I_{ff,1}(x) &=&\iint \left\{ \rho \left( x-\tau u_{1}\right) -\rho \left(
x-\tau u_{2}\right) \right\} \\ &&\qquad \times \left( \frac{\tau ^{2}}{2}f_{\xi ,X^{\ast
}|Z}^{\prime \prime }\left( u_{1},x|z_{1}\right) u_{1}^{2} - \int_{0}^{\tau }%
\frac{1}{2}f_{\xi ,X^{\ast }|Z}^{\prime \prime \prime }\left(
u_{1},x-tu_{1}|z_{1}\right) u_{1}^{3}\left( \tau -t\right) ^{2}dt\right)  \\
&&\qquad \times \left( -f_{\xi ,X^{\ast }|Z}^{\prime }\left(
u_{2},x|z_{2}\right) \tau u_{2}+f_{\xi ,X^{\ast }|Z}\left(
u_{2},x|z_{2}\right) \right) du_{1}du_{2} \\
&=&\frac{\tau ^{2}}{2}\left( -\tau I_{\Delta \left( x,2,1\right) }+I_{\Delta
\left( x,2,0\right) }\right) +I_{ff,1,2}(x),
\end{eqnarray*}%
where the remainder is%
\begin{eqnarray*}
I_{ff,1,2}(x) &=& - \iint \left\{ \rho \left( x-\tau u_{1}\right) -\rho \left(
x-\tau u_{2}\right) \right\} \left( \int_{0}^{\tau }\frac{1}{2}f_{\xi
,X^{\ast }|Z}^{\prime \prime \prime }\left( u_{1},x-tu_{1}|z_{1}\right)
u_{1}^{3}\left( \tau -t\right) ^{2}dt\right)  \\
&&\qquad \times \left( -f_{\xi ,X^{\ast }|Z}^{\prime }\left(
u_{2},x|z_{2}\right) \tau u_{2}+f_{\xi ,X^{\ast }|Z}\left(
u_{2},x|z_{2}\right) \right) du_{1}du_{2}.
\end{eqnarray*}%
By an argument similar to the one provided in part 3.1 above,
\begin{eqnarray*}
I_{ff,1,2,1}(x) &=& \iint \left\{ \rho \left( x-\tau u_{1}\right) -\rho \left(
x-\tau u_{2}\right) \right\} \left( \int_{0}^{\tau }\frac{1}{2}f_{\xi
,X^{\ast }|Z}^{\prime \prime \prime }\left( u_{1},x-tu_{1}|z_{1}\right)
u_{1}^{3}\left( \tau -t\right) ^{2}dt\right)  \\
&&\qquad \times f_{\xi ,X^{\ast }|Z}^{\prime }\left(
u_{2},x|z_{2}\right) \tau u_{2} du_{1}du_{2} \\
&=&O_x\left(  \tau ^{4}\right).
\end{eqnarray*}%
Next,
\begin{eqnarray*}
I_{ff,1,2,2}(x) &=& - \iint \left\{ \rho \left( x-\tau u_{1}\right) -\rho \left(
x-\tau u_{2}\right) \right\} \left( \int_{0}^{\tau }\frac{1}{2}f_{\xi
,X^{\ast }|Z}^{\prime \prime \prime }\left( u_{1},x-tu_{1}|z_{1}\right)
u_{1}^{3}\left( \tau -t\right) ^{2}dt\right)  \\
&&\qquad \times f_{\xi ,X^{\ast }|Z}\left(
u_{2},x|z_{2}\right) du_{1}du_{2} \\
&=& - \iint \left( \sum_{j=1}^2 \left(-1\right)^j \int_{0}^{\tau} \rho'(x - t u_j ) u_j dt \right) \left( \int_{0}^{\tau }\frac{1}{2}f_{\xi
,X^{\ast }|Z}^{\prime \prime \prime }\left( u_{1},x-tu_{1}|z_{1}\right)
u_{1}^{3}\left( \tau -t\right) ^{2}dt\right)  \\
&&\qquad \times f_{\xi ,X^{\ast }|Z}\left(
u_{2},x|z_{2}\right) du_{1}du_{2} \\
&=&O_x\left(  \tau ^{4}\right),
\end{eqnarray*}%
where the second equality follows from equation~(\ref{eq:pf:lem:diff q: d rho
approx ord 1}), and the last follows from an argument similar to the one provided in part 3.1. Then, we conclude
\begin{align*}
	I_{ff,1,2}(x) = I_{ff,1,2,1}(x) + I_{ff,1,2,2}(x) = O_x (\tau^4).
\end{align*}

Next, note that $I_{h\left( x,0,1,2,0\right) }=0$ and $I_{h\left( x,1,0,2,0\right) }=0$ \ since $\mathbb{M}_{|\zeta
_{1}\zeta _{2}z_{1}z_{2}}\left[ U_{1}^{2}U_{2}^{{}}\right] =0$ and $\mathbb{M%
}_{|\zeta _{1}\zeta _{2}z_{1}z_{2}}\left[ U_{1}^{3}\right] =0$,
respectively. Thus, using equation~(\ref{eq:pf:lem:diff q: int psi approx
ord 2})%
\begin{eqnarray*}
I_{ff,1}(x) &=&\frac{\tau ^{2}}{2}\left( -\tau I_{\Delta \left( x,2,1\right)
}+I_{\Delta \left( x,2,0\right) }\right) +O_x\left( 
\tau ^{4}\right) =\tau ^{3}O_x\left(  \tau \right) +%
\frac{\tau ^{2}}{2}I_{\Delta \left( x,2,0\right) }+O_x\left(  \tau ^{4}\right)  \\
&=&\frac{\tau ^{2}}{2}\left( \tau \rho ^{\prime }\left( x\right) \left\{
I_{h\left( x,0,1,2,0\right) }-I_{h\left( x,1,0,2,0\right) }\right\} +O_x\left( \tau ^{2}\right) \right) +O_x\left(  \tau ^{4}\right)  \\
&=&O_x\left(  \tau ^{4}\right).
\end{eqnarray*}

Similarly, 
\begin{equation*}
I_{ff,2}\left( x\right) =O_x\left(  \tau ^{4}\right) .
\end{equation*}

\bigskip

\noindent \textbf{3.3} The main term is 
\begin{eqnarray*}
I_{ff,4}\left( x\right)  &\equiv &\iint \left\{ \rho \left( x-\tau
u_{1}\right) -\rho \left( x-\tau u_{2}\right) \right\} \left( -f_{\xi
,X^{\ast }|Z}^{\prime }\left( u_{1},x|z_{1}\right) \tau u_{1}+f_{\xi
,X^{\ast }|Z}\left( u_{1},x|z_{1}\right) \right)  \\
&&\qquad \times \left( -f_{\xi ,X^{\ast }|Z}^{\prime }\left(
u_{2},x|z_{2}\right) \tau u_{2}+f_{\xi ,X^{\ast }|Z}\left(
u_{2},x|z_{2}\right) \right) du_{1}du_{2} \\
&=&\tau ^{2}I_{\Delta \left( x,1,1\right) }-\tau I_{\Delta \left(
x,1,0\right) }-\tau I_{\Delta \left( x,0,1\right) }+I_{\Delta \left(
x,0,0\right) } \\
&=&-\tau \left( I_{\Delta \left( x,1,0\right) }+I_{\Delta \left(
x,0,1\right) }\right) ,
\end{eqnarray*}%
where the last equality holds because $I_{\Delta \left( x,0,0\right)
}=I_{\Delta \left( x,1,1\right) }=0$ due to the symmetry. Using equation~(%
\ref{eq:pf:lem:diff q: int psi approx ord 3}) we have%
\begin{equation*}
I_{\Delta \left( x,1,0\right) }=\tau \rho ^{\prime }\left( x\right) \left\{
I_{h\left( x,0,1,1,0\right) }-I_{h\left( x,1,0,1,0\right) }\right\} +\frac{%
\tau ^{2}}{2}\rho ^{\prime \prime }\left( x\right) \left\{ I_{h\left(
x,2,0,1,0\right) }-I_{h\left( x,0,2,1,0\right) }\right\} +O_x\left(
 \tau ^{3}\right) .
\end{equation*}%
Using equation~(\ref{eq:pf:lem:diff q: int h}),%
\begin{eqnarray*}
I_{h\left( x,0,1,1,0\right) } &=&\left. \nabla _{\zeta _{1}}\mathbb{M}%
_{|\zeta _{1}\zeta _{2}z_{1}z_{2}}\left[ U_{1}U_{2}\right] \right\vert
_{\zeta _{1}=\zeta _{2}=x}=0, \\
I_{h\left( x,1,0,1,0\right) } &=&\left. \nabla _{\zeta _{1}}\mathbb{M}%
_{|\zeta _{1}\zeta _{2}z_{1}z_{2}}\left[ U_{1}^{2}\right] \right\vert
_{\zeta _{1}=\zeta _{2}=x}=\left. \nabla _{\zeta _{1}}\left( E[\xi_i^2|X_i^* = \zeta_1] f_{X^{\ast }|Z}\left( \zeta _{1}|z_{1}\right) f_{X^{\ast
}|Z}\left( \zeta _{2}|z_{2}\right) \right) \right\vert _{\zeta _{1}=\zeta
_{2}=x} \\
&=&\left[E[\xi_i^2|X_i^* = x] f_{X^{\ast }|Z}^{\prime }\left( x|z_{1}\right)
+ \nabla_x E[\xi_i^2|X_i^* = x] f_{X^{\ast }|Z}\left( x|z_{1}\right)\right] f_{X^{\ast
}|Z}\left( x|z_{2}\right),  \\
I_{h\left( x,2,0,1,0\right) } &=&\left. \nabla _{\zeta _{1}}\mathbb{M}%
_{|\zeta _{1}\zeta _{2}z_{1}z_{2}}\left[ U_{1}^{3}\right] \right\vert
_{\zeta _{1}=\zeta _{2}=x}=\left. \nabla _{\zeta _{1}}0\right\vert _{\zeta
_{1}=\zeta _{2}=x}=0, \\
I_{h\left( x,0,2,1,0\right) } &=&\left. \nabla _{\zeta _{1}}\mathbb{M}%
_{|\zeta _{1}\zeta _{2}z_{1}z_{2}}\left[ U_{1}U_{2}^{2}\right] \right\vert
_{\zeta _{1}=\zeta _{2}=x}=0,
\end{eqnarray*}%
and hence 
\begin{equation*}
I_{\Delta \left( x,1,0\right) }=-\tau \rho ^{\prime }\left( x\right) \left[
E[\xi_i^2|X_i^* = x] f_{X^{\ast }|Z}^{\prime }\left( x|z_{1}\right) + \nabla_x E[\xi_i^2|X_i^* = x] f_{X^{\ast }|Z}\left( x|z_{1}\right) \right] f_{X^{\ast
}|Z}\left( x|z_{2}\right) +O_x\left(  \tau ^{3}\right) .
\end{equation*}%
Similarly,%
\begin{equation*}
I_{\Delta \left( x,0,1\right) }=\tau \rho ^{\prime }\left( x\right) \left[
E[\xi_i^2|X_i^* = x] f_{X^{\ast }|Z}^{\prime }\left( x|z_{2}\right) + \nabla_x E[\xi_i^2|X_i^* = x] f_{X^{\ast }|Z}\left( x|z_{2}\right) \right] f_{X^{\ast
}|Z}\left( x|z_{1}\right) +O_x\left(  \tau ^{3}\right) ,
\end{equation*}%
and hence%
\begin{eqnarray*}
I_{ff,4}\left( x\right)  &=&-\tau \left( I_{\Delta \left( x,1,0\right)
}+I_{\Delta \left( x,0,1\right) }\right)  \\
&=& \rho ^{\prime }\left( x\right) v\left( x\right) \left( f_{X^{\ast
}|Z}^{\prime }\left( x|z_{1}\right) f_{X^{\ast }|Z}\left( x|z_{2}\right)
-f_{X^{\ast }|Z}\left( x|z_{1}\right) f_{X^{\ast }|Z}^{\prime }\left(
x|z_{2}\right) \right) +O_x\left(  \tau ^{4}\right),
\end{eqnarray*}
where we used $v(x) = \tau^2 E[\xi_i^2|X_i^* = x]$.

\noindent \textbf{4}. Putting all $I_{ff,j}\left( x\right) $ together we have%
\begin{eqnarray*}
&&N\left( x,z_{1},z_{2},\tau \right) \\
&=&\rho ^{\prime }\left( x\right) v\left( x\right) \left(
f_{X^{\ast }|Z}^{\prime }\left( x|z_{1}\right) f_{X^{\ast }|Z}\left(
x|z_{2}\right) -f_{X^{\ast }|Z}\left( x|z_{1}\right) f_{X^{\ast }|Z}^{\prime
}\left( x|z_{2}\right) \right) +O_x\left(  \tau
^{4}\right) \\
&=&\rho ^{\prime }\left( x\right) v\left( x\right) \left(
s_{X^{\ast }|Z}\left( x|z_{1}\right) -s_{X^{\ast }|Z}\left( x|z_{2}\right) \right) f_{X^{\ast }|Z}\left( x|z_{1}\right)
f_{X^{\ast }|Z}\left( x|z_{2}\right) +O_x\left(  \tau
^{4}\right) \text{.}
\end{eqnarray*}

Moreover, Lemma \ref{lem:smooth q prime and s prime} establishes $f_{X|Z}(x,z_j) = f_{X^*|Z} (x|z_j) + O_x (\tau^2)$ for $j \in \{1,2\}$. Hence, the denominator $D\left( x,z_{1},z_{2},\tau \right) =f_{X|Z}\left( x|z_{1}\right) f_{X|Z}\left( x|z_{2}\right) $ satisfies
\begin{equation*}
D\left( x,z_{1},z_{2},\tau \right) =f_{X^{\ast }|Z}\left( x|z_{1}\right)
f_{X^{\ast }|Z}\left( x|z_{2}\right) + O_x (\tau^2).
\end{equation*}%
Thus,%
\begin{eqnarray*}
&&\frac{N\left( x,z_{1},z_{2},\tau \right) }{D\left( x,z_{1},z_{2},\tau
\right) } \\
&=&\frac{\rho ^{\prime }\left( x\right) v\left( x\right) f_{X^{\ast
}|Z}\left( x|z_{1}\right) f_{X^{\ast }|Z}\left( x|z_{2}\right) \left(
s_{X^{\ast }|Z}\left( x|z_{1}\right) -s_{X^{\ast }|Z}\left( x|z_{2}\right) \right) +O_x\left(  \tau
^{4}\right) }{f_{X^{\ast }|Z}\left( x|z_{1}\right) f_{X^{\ast }|Z}\left(
x|z_{2}\right) +O_x\left(  \tau ^{2}\right) } \\
&=&\rho ^{\prime }\left( x\right) v\left( x\right) \left(
s_{X^{\ast }|Z}\left( x|z_{1}\right) -s_{X^{\ast }|Z}\left( x|z_{2}\right) \right) +O_x\left(  \tau
^{4}\right) ,
\end{eqnarray*}%
which completes the proof.
\end{proof}

\subsubsection{Proof of the main result}
Equipped with Lemmas \ref{lem:q(x,z) expansion}-\ref{lem:diff q}, we are ready to prove Theorem \ref{thm:NPID-non-cl:IV}.

\begin{proof}[Proof of Theorem \ref{thm:NPID-non-cl:IV}]
	Using equation~(\ref{eq:NPID:q(x,z) K2 exp}) we have%
	\begin{equation}
	q\left( x,z_{1}\right) -q\left( x,z_{2}\right) =v\left( x\right)
	\rho ^{\prime }\left( x\right) \left[ s_{X^{\ast }|Z}\left( x|z_{1}\right)
	-s_{X^{\ast }|Z}\left( x|z_{2}\right) \right] +O\left( \tau ^{p}\right) . \label{eq:pf:delta q(x,z)}
	\end{equation}%
	Combined with \eqref{eq:pf: s q prime bias}, the above implies that%
	\begin{equation*}
	q\left( x,z_{1}\right) -q\left( x,z_{2}\right) =v\left( x\right)
	q^{\prime }\left( x\right) \left[ s_{X|Z}\left( x|z_{1}\right)
	-s_{X|Z}\left( x|z_{2}\right) \right] +O\left( \tau ^{p}\right) .
	\end{equation*}

	Hence, for any $x$ with $\rho ^{\prime }\left( x\right) \left[ s_{X^{\ast
	}|Z}\left( x|z_{1}\right) -s_{X^{\ast }|Z}\left( x|z_{2}\right) \right] \neq
	0$ (or $x$ such that $\left\vert q^{\prime }\left( x\right) \left[
	s_{X|Z}\left( x|z_{1}\right) -s_{X|Z}\left( x|z_{2}\right) \right]
	\right\vert >C>0$),%
	\begin{equation}
	\widetilde{v}\left( x\right) =v\left( x\right) +O\left( \tau
	^{p}\right) ,\text{\qquad where }\widetilde{v}\left( x\right) \equiv \frac{%
	q\left( x,z_{1}\right) -q\left( x,z_{2}\right) }{q^{\prime }\left( x\right) %
	\left[ s_{X|Z}\left( x|z_{1}\right) -s_{X|Z}\left( x|z_{2}\right) \right] }.
	\label{eq:pf:v-tilde result}
	\end{equation}

	Next by Lemma~\ref{lem:diff q}, we have
	\begin{equation*}
	q\left( x,z_{1}\right) -q\left( x,z_{2}\right) =\tau ^{2}\rho ^{\prime
	}\left( x\right) v\left( x\right) \left( s_{X^{\ast }|Z}\left(
	x|z_{1}\right) -s_{X^{\ast }|Z}\left( x|z_{2}\right) \right)
	+O_x\left( \tau ^{p}\right).
	\end{equation*}%
	By Lemma \ref{lem:smooth q prime and s prime}, we also have $q'(x) = \rho'(x) + O_x(\tau^{p-2})$ and $s_{X|Z}\left(x|z\right) = s_{X^{\ast }|Z}\left(x|z\right) + O_x(\tau^{p-2})$. Hence,
	\begin{equation*}
	q\left( x,z_{1}\right) -q\left( x,z_{2}\right) = q ^{\prime
	}\left( x\right) v\left( x\right) \left( s_{X|Z}\left(
	x|z_{1}\right) -s_{X|Z}\left( x|z_{2}\right) \right)
	+O_x\left( \tau ^{p}\right).
	\end{equation*}%
	This implies that%
	\begin{equation*}
	q^{\prime }\left( x,z_{1}\right) -q^{\prime }\left( x,z_{2}\right) = \nabla _{x}\left[ q ^{\prime }\left( x\right) v\left( x\right) \left(
	s_{X^|Z}\left( x|z_{1}\right) -s_{X^|Z}\left( x|z_{2}\right) \right) \right] +O\left( \tau ^{p}\right) .
	\end{equation*}
	Hence, $v'(x)$ satisfies
	\begin{align*}
		v'(x) &= \frac{q^{\prime }\left( x,z_{1}\right)
	-q^{\prime }\left( x,z_{2}\right) }{q^{\prime }\left( x\right) \left[
	s_{X|Z}\left( x|z_{1}\right) -s_{X|Z}\left( x|z_{2}\right) \right] }-{v%
	}\left( x\right) \frac{\nabla _{x}\left( q^{\prime }\left( x\right) \left[
	s_{X|Z}\left( x|z_{1}\right) -s_{X|Z}\left( x|z_{2}\right) \right] \right) }{%
	q^{\prime }\left( x\right) \left[ s_{X|Z}\left( x|z_{1}\right)
	-s_{X|Z}\left( x|z_{2}\right) \right] } + O (\tau^p) \\
	&=\underbrace{\frac{q^{\prime }\left( x,z_{1}\right)
	-q^{\prime }\left( x,z_{2}\right) }{q^{\prime }\left( x\right) \left[
	s_{X|Z}\left( x|z_{1}\right) -s_{X|Z}\left( x|z_{2}\right) \right] }-\widetilde {v%
	}\left( x\right) \frac{\nabla _{x}\left( q^{\prime }\left( x\right) \left[
	s_{X|Z}\left( x|z_{1}\right) -s_{X|Z}\left( x|z_{2}\right) \right] \right) }{%
	q^{\prime }\left( x\right) \left[ s_{X|Z}\left( x|z_{1}\right)
	-s_{X|Z}\left( x|z_{2}\right) \right] }}_{=\widetilde v'(x)} + O (\tau^p),
	\end{align*}
	where the second equality uses equation \eqref{eq:pf:v-tilde result}. Hence, we conclude
	\begin{align}
	\label{eq:pf:v-tilde prime result}
	\widetilde{v}^{\prime }\left( x\right) =v^{\prime }\left( x\right) +O\left( \tau
	^{p}\right).
	\end{align}

	Next, notice that since $q'(x) = \rho'(x) + O_x (\tau^{p-2})$ (Lemma \ref{lem:smooth q prime and s prime}), we also have
	\begin{align}
		\label{eq:pf:q'' bias}
		\abs{q''(x) - \rho''(x)} = O (\tau^{p-2}).
	\end{align}

	Finally, recall that \eqref{eq:NPID:q(x,z) K2 exp} implies
	\begin{align*}
		\rho(x) &= q(x,z_1) - {v}\left(
		x\right) \left[ \rho^{\prime }\left( x\right) s_{X^*|Z}\left( x|z\right) +\tfrac{1%
		}{2}\rho^{\prime \prime }\left( x\right) \right] - {v}^{\prime }\left(
		x\right) \rho^{\prime }\left( x\right) + O(\tau^p),\\
		&= q(x,z_1) - {v}\left(
		x\right) \left[ q^{\prime }\left( x\right) s_{X|Z}\left( x|z\right) +\tfrac{1%
		}{2}q^{\prime \prime }\left( x\right) \right] - {v}^{\prime }\left(
		x\right) q^{\prime }\left( x\right) + O(\tau^p) \\
		&= q(x,z_1) - \widetilde {v}\left(
		x\right) \left[ q^{\prime }\left( x\right) s_{X|Z}\left( x|z\right) +\tfrac{1%
		}{2}q^{\prime \prime }\left( x\right) \right] - \widetilde {v}^{\prime }\left(
		x\right) q^{\prime }\left( x\right) + O(\tau^p)\\
		&= \widetilde \rho (x,z_1) + O(\tau^p),	
	\end{align*}
	where the second equality uses $v(x) = O (\tau^2)$, $v'(x) = O(\tau^2)$, and equations \eqref{eq:pf: s q prime bias} and \eqref{eq:pf:q'' bias}, and the third equality uses \eqref{eq:pf:v-tilde result} and \eqref{eq:pf:v-tilde prime result}.
\end{proof}

\subsection{Proofs of Lemma \ref{lem:NCME:IV:rates}}

\begin{proof}[Proof of Lemma \ref{lem:NCME:IV:rates}]
	First, consider $\int \eta(r) f_{\varepsilon|X^*} (x - r|r) dr =\int \eta (x - \tau u) f_{\xi|X^*} (u|x - \tau u) du$ where $\eta (\cdot)$ is a bounded function with $m \geq p$ bounded derivatives with respect to $x$. Then, notice that Assumptions \ref{ass:NPID:smoothness} and \ref{ass:NPID:dominance} ensure that $\int \eta (x - \tau u) f_{\xi|X^*} (u|x - \tau u) du$ has $m$ bounded derivatives with respect to $x$. Hence, $f_{X|Z}(x|z) = \int f_{X^*|Z}(r|z) f_{\varepsilon|X^*} (x - r|r) dr$, $f_X(x) = \int f_{X^*}(r) f_{\varepsilon|X^*} (x - r|r) dr$, $q(x,z) = \int \rho(r) f_{X^*|Z}(r|z) f_{\varepsilon|X^*} (x - r|r) dr/f_{X|Z}(x|z)$, and $q(x) = \int \rho(r) f_{X^*}(r) f_{\varepsilon|X^*} (x - r|r) dr/f_X(x)$ have 4 bounded derivatives with respect to $x$ (for $q(x,z)$ and $q(x)$, at least, in a neighborhood of $x$ where $f_{X^*|Z}(x|z)$ and $f_{X^*}(z)$ are bounded away from zero).

	If the tuning parameters are chosen optimally, under the usual regularity
	conditions, the rates of convergence of these estimators are $\widehat{q}%
	\left( x,z\right) -q\left( x,z\right) =O_{p}\left( n^{-\frac{m}{2m+1}%
	}\right) $, $\widehat{q}\left( x\right) -q\left( x\right) =O_{p}\left( n^{-%
	\frac{m}{2m+1}}\right) $, $\widehat{q}'%
	\left( x,z\right) -q'\left( x,z\right) =O_{p}\left( n^{-\frac{m-1}{2m+1}%
	}\right) $, $\widehat{q}^{\prime }\left( x\right) -q^{\prime
	}\left( x\right) =O_{p}\left( n^{-\frac{m-1}{2m+1}}\right) $, $\widehat{q}%
	^{\prime \prime }\left( x\right) -q^{\prime \prime }\left( x\right)
	=O_{p}\left( n^{-\frac{m-2}{2m+1}}\right) $, $\widehat{s}_{X|Z}\left(
	x|z\right) -s_{X|Z}\left( x|z\right) =O_{p}\left( n^{-\frac{m-1}{2m+1}%
	}\right) $, and $\widehat{s}_{X|Z}'\left(
	x|z\right) -s_{X|Z}'\left( x|z\right) =O_{p}\left( n^{-\frac{m-2}{2m+1}%
	}\right) $ for $x\in S_{X^{\ast }}\left( z\right) $ where $\widehat{s}%
	_{X|Z}\left( x|z\right) \equiv \left. \widehat{f}_{X|Z}^{\prime }\left(
	x|z\right) \right/ \widehat{f}_{X|Z}\left( x|z\right) $ (see, e.g., \citealp{Stone1980AoS}).

	First, note that
	\begin{align}
		\widehat {v}\left( x\right) &= \frac{\widehat q\left( x,z_{1}\right) - \widehat q\left(
		x,z_{2}\right) }{\widehat q^{\prime }\left( x\right) \left[\widehat s_{X|Z}\left(
		x|z_{1}\right) - \widehat s_{X|Z}\left( x|z_{2}\right) \right] } \notag \\
		& = \underbrace{\frac{ q\left( x,z_{1}\right) -  q\left(
		x,z_{2}\right) }{ q^{\prime }\left( x\right) \left[ s_{X|Z}\left(
		x|z_{1}\right) -  s_{X|Z}\left( x|z_{2}\right) \right] }}_{= \widetilde{v}(x)} +O_{p}\left( n^{-\frac{m}{2m+1}}+\tau _{n}^{2}n^{-\frac{%
		m-1}{2m+1}}\right) \label{eq:pf:v hat rate}
	\end{align}
	Here the second equality follows from $\widehat{a}/\widehat{b}-a/b=\left( \widehat{a}-a\right) /\widehat{b}%
	+a(1/\widehat{b}-1/b)$, $q\left( x,z_{1}\right)
	-q\left( x,z_{2}\right) =O\left( \tau _{n}^{2}\right) $ (implied by \eqref{eq:pf:delta q(x,z)}), and the rates
	of convergence listed above.%

	Next,
	\begin{align*}
		\widehat v' (x) = \frac{\widehat q^{\prime }\left( x,z_{1}\right)
	- \widehat q^{\prime }\left( x,z_{2}\right) }{\widehat q^{\prime }\left( x\right) \left[
	\widehat s_{X|Z}\left( x|z_{1}\right) -\widehat s_{X|Z}\left( x|z_{2}\right) \right] }-{\widehat v%
	}\left( x\right) \frac{\nabla _{x}\left(\widehat q^{\prime }\left( x\right) \left[
	\widehat s_{X|Z}\left( x|z_{1}\right) -\widehat s_{X|Z}\left( x|z_{2}\right) \right] \right) }{%
	\widehat q^{\prime }\left( x\right) \left[\widehat s_{X|Z}\left( x|z_{1}\right)
	-\widehat s_{X|Z}\left( x|z_{2}\right) \right] }.
	\end{align*}
	Analogously to the derivation provided above, we inspect that
	\begin{align*}
		\frac{\widehat q^{\prime }\left( x,z_{1}\right)
	- \widehat q^{\prime }\left( x,z_{2}\right) }{\widehat q^{\prime }\left( x\right) \left[
	\widehat s_{X|Z}\left( x|z_{1}\right) -\widehat s_{X|Z}\left( x|z_{2}\right) \right] } = 
	\frac{q^{\prime }\left( x,z_{1}\right)
	-q^{\prime }\left( x,z_{2}\right) }{q^{\prime }\left( x\right) \left[
	s_{X|Z}\left( x|z_{1}\right) -s_{X|Z}\left( x|z_{2}\right) \right] } + O_{p}\left(n^{-\frac{m-1}{2m+1}}\right),
	\end{align*}
	and
	\begin{align*}
		\frac{\nabla _{x}\left(\widehat q^{\prime }\left( x\right) \left[
	\widehat s_{X|Z}\left( x|z_{1}\right) -\widehat s_{X|Z}\left( x|z_{2}\right) \right] \right) }{%
	\widehat q^{\prime }\left( x\right) \left[\widehat s_{X|Z}\left( x|z_{1}\right)
	-\widehat s_{X|Z}\left( x|z_{2}\right) \right] } = \frac{\nabla _{x}\left( q^{\prime }\left( x\right) \left[
	 s_{X|Z}\left( x|z_{1}\right) - s_{X|Z}\left( x|z_{2}\right) \right] \right) }{%
	 q^{\prime }\left( x\right) \left[ s_{X|Z}\left( x|z_{1}\right)
	- s_{X|Z}\left( x|z_{2}\right) \right] } +  O_{p}\left(n^{-\frac{m-2}{2m+1}}\right).
	\end{align*}
	Combining the latter with \eqref{eq:pf:v hat rate}, we obtain
	\begin{align*}
		{\widehat v%
	}\left( x\right) \frac{\nabla _{x}\left(\widehat q^{\prime }\left( x\right) \left[
	\widehat s_{X|Z}\left( x|z_{1}\right) -\widehat s_{X|Z}\left( x|z_{2}\right) \right] \right) }{%
	\widehat q^{\prime }\left( x\right) \left[\widehat s_{X|Z}\left( x|z_{1}\right)
	-\widehat s_{X|Z}\left( x|z_{2}\right) \right] } = &{ \widetilde v%
	}\left( x\right) \frac{\nabla _{x}\left( q^{\prime }\left( x\right) \left[
	 s_{X|Z}\left( x|z_{1}\right) - s_{X|Z}\left( x|z_{2}\right) \right] \right) }{%
	 q^{\prime }\left( x\right) \left[ s_{X|Z}\left( x|z_{1}\right)
	- s_{X|Z}\left( x|z_{2}\right) \right] } \\
	&+O_{p}\left( n^{-\frac{m}{2m+1}}+\tau _{n}^{2}n^{-\frac{m-2}{2m+1}}\right),
	\end{align*}
	where we also used $\widetilde v(x) = O(\tau_n^2)$. Hence, we conclude
	\begin{align}
		\widehat v'(x) = \widetilde v'(x) + O_{p}\left( n^{-\frac{m-1}{2m+1}}+\tau _{n}^{2}n^{-\frac{m-2}{2m+1}}\right).\label{eq:pf:v hat prime rate}
	\end{align}
	Finally,
	\begin{align*}
		\widehat \rho (x) 
		&= \widehat q(x,z_1) - \widehat {v}\left(
		x\right) \left[ \widehat q^{\prime }\left( x\right) \widehat s_{X|Z}\left( x|z_1\right) +\tfrac{1%
		}{2}\widehat q^{\prime \prime }\left( x\right) \right] - \widehat {v}^{\prime }\left(
		x\right) \widehat q^{\prime }\left( x\right)\\
		&= \widetilde \rho (x,z_1) + O_{p}\left( n^{-\frac{m-1}{2m+1}}+\tau _{n}^{2}n^{-\frac{m-2}{2m+1}}\right) \\
		&= \rho(x) +  O_{p}\left( n^{-\frac{m-1}{2m+1}}+\tau _{n}^{2}n^{-\frac{m-2}{2m+1}} + \tau_n^p\right),\\
		&= \rho(x) +  O_{p}\left( n^{-\frac{m-1}{2m+1}} + \tau_n^p\right).
	\end{align*}
	Here, the second equality uses the definition of $\widetilde \rho (x,z_1)$, and equations \eqref{eq:pf:v hat rate} and \eqref{eq:pf:v hat prime rate}. The third equality follow from Theorem \ref{thm:NPID-non-cl:IV}. For the last equality, note that $\tau _{n}^{2}n^{-\frac{m-2}{2m+1}} \lesssim n^{-\frac{m-1}{2m+1}}$ for $\tau_n \lesssim
 	n^{-\frac{1}{2(2m+1)}}$ and $\tau _{n}^{2}n^{-\frac{m-2}{2m+1}} \lesssim \tau_n^p$ for $\tau_n \gtrsim n^{-\frac{m-2}{(p-2)(2m+1)}}$, and $n^{-\frac{1}{2(2m+1)}} > n^{-\frac{m-2}{(p-2)(2m+1)}}$ since $2m - p - 2 > 0$.%

 	For the naive estimator, we have
	\begin{equation*}
	\widehat{\rho }^{\text{Naive}}\left( x\right) =\widehat{q}\left( x\right)
	=\rho \left( x\right) +O_{p}\left( n^{-\frac{m}{2m+1}} + \tau_{n}^{2}\right),
	\end{equation*}%
	where the second equality uses $\widehat q(x) = q(x) + O \left(n^{-\frac{m}{2m+1}}\right)$ and $q (x) = \rho (x) + O(\tau_n^2)$ (e.g., \eqref{eq:NPID:q(x,z) K2 exp}).%
\end{proof}	

\subsection{Proof of Theorem \ref{thm:NCME rho tilde}}

Before proving Theorem \ref{thm:NCME rho tilde}, in Section \ref{ssec:A:CDF and QF ID}, we first demonstrate that the CDF $F_{X^*}(\cdot)$ and the quantile function $Q_{X^*} (\cdot)$ of $X_{i}^*$ are identified up to an error of order $O(\tau^p)$. The proof of Theorem \ref{thm:NCME rho tilde} is then provided in Section \ref{ssec:A:NCME ID proof}.

\subsubsection{Identification of $F_{X^*} (\cdot)$ and $Q_{X^*} (\cdot)$}
\label{ssec:A:CDF and QF ID}

\begin{lemma}
\label{lem:CDF and QF} 
Suppose that Assumptions \ref{ass:NCME}-\ref{ass:NPID:dominance} and \ref%
{ass:NPID:integrable derivatives} are satisfied. Suppose either (i) $p=3$, or (ii) $E\left[ \xi _{i}^{3}|X_{i}^{\ast }%
\right] =0$ and $p=4$. Then, for any $x \in \supp{X_i^*}$ and for any $s \in
(0,1)$ such that $f_{X^*}(Q_{X^*}(s)) > 0$, as $\tau \rightarrow 0$,
\begin{eqnarray*}
F_{X^{\ast }}\left( x\right) &=&F_{X}\left( x\right) -\frac{1}{2}\nabla
_{x}\left( f_{X} (x) v\left( x\right) \right) +O\left( \tau ^{p}\right), \\
Q_{X^{\ast }}\left( s\right) &=&Q_{X}\left( s\right) +\frac{1}{2}\left\{
s_{X}\left( Q_{X}\left( s\right) \right) v\left( Q_{X}\left( s\right)
\right) +\nabla _{x} v\left( Q_{X}\left( s\right) \right) \right\} +O\left(
\tau ^{p}\right).
\end{eqnarray*}
\end{lemma}

\begin{remark}
Assumptions \ref{ass:NCME} and \ref{ass:NPID:smoothness} are stronger than
needed for the result of Lemma \ref{lem:CDF and QF} to hold. For Assumption %
\ref{ass:NCME}, the requirement $\xi \perp Z_i|X_i^*$ can be dropped.
Assumption \ref{ass:NPID:smoothness} can be replaced by requiring the
densities $f_{X^*} (x)$ and $f_{\xi|X^*} (\xi|x)$ to be bounded
functions that have at least $p \geqslant 3$ bounded derivatives (all with
respect to $x$).
\end{remark}

Lemma \ref{lem:CDF and QF} extends the results obtained earlier for CME (e.g., Chesher,
1991) to the WCME\ model. It also provides a way of recovering $F_{X^*}(x)$ and $%
Q_{X^*}(s)$ up to an error of order $O(\tau^p)$ using 
\begin{eqnarray*}
\tilde F_{X^{\ast }}\left( x\right) &\equiv&F_{X}\left( x\right) - \frac{1}{2%
}\nabla _{x}\left( f_{X} (x) \tilde v\left( x\right) \right), \\
\tilde Q_{X^{\ast }}\left( s\right) &\equiv&Q_{X}\left( s\right) +\frac{1}{2}%
\left\{ s_{X}\left( Q_{X}\left( s\right) \right) \tilde v\left( Q_{X}\left(
s\right) \right) +\nabla _{x} \tilde v\left( Q_{X}\left( s\right) \right)
\right\},
\end{eqnarray*}
where $\tilde v (\cdot)$ is given by \eqref{eq: v tilde}.

\begin{proof}[Proof of Lemma \ref{lem:CDF and QF}]
	In the proof of Lemma \ref{lem:q(x,z) expansion}, we demonstrated that, for a bounded function $\eta(\cdot)$ with $p$ bounded derivatives, we have%
	\begin{align*}
		\int \eta \left( r\right) f_{\varepsilon |X^{\ast }}\left( x-r|r\right)
		dr &= \int \eta \left( x-\tau u\right) f_{\xi |X^{\ast }}\left( u|x-\tau
		u\right) du \\
		&= \eta \left( x\right) + \frac{1}{2} \nabla_x^2 \left\{\eta \left( x\right) v\left( x\right) \right\} + R_\eta (x; \tau).
	\end{align*}
	Here the remainder $R_\eta (x; \tau)$ can be represented in the integral form as
	\begin{align*}
		R_{\eta} (x;\tau) = \int \left(\int_{0}^\tau \frac{(-1)^p}{(p-1)!} u^p \nabla_x^p \{\eta(x - tu) f_{\xi|X^*}(u|x - tu)\} (\tau - t)^{p-1} dt\right) du.%
	\end{align*}
	Then, for $f_{X}\left( x\right) =\int f_{X^{\ast }}\left( r\right) f_{\varepsilon |X^{\ast }}\left( x-r|r\right) dr$ with $\eta (\cdot) = f_{X^*} (\cdot) $, we have
	\begin{equation*}
		f_{X}\left( x\right) =f_{X^{\ast }}\left( x\right) +\frac{1}{2}\nabla
		_{x}^{2}\left( f_{X^{\ast }}^{{}}\left( x\right) v\left( x\right)
		\right) + R_{f_{X^*}}(x; \tau).
	\end{equation*}%
	Then,
	\begin{align*}
		F_{X}(\bar x) &= \int_{- \infty}^{\bar x} f_X (x) dx \\
		&= \int_{- \infty}^{\bar x} \left(f_{X^{\ast }}\left( x\right) +\frac{1}{2}\nabla
						_{x}^{2}\left( f_{X^{\ast }}^{{}}\left( x\right) v\left( x\right)
						\right) + R_{f_{X^*}}(x; \tau)\right) dx\\
		&= F_{X^*}(x) + \frac{1}{2} \nabla_x \left( f_{X^{\ast }}^{{}}\left( \bar x\right) v\left( \bar x\right)
		\right) + \int_{- \infty}^{\bar x} R_{f_{X^*}}(x; \tau) dx.
	\end{align*}
	Here we used
	\begin{align*}
		\int_{- \infty}^{\bar x} \nabla_{x}^{2}\left( f_{X^{\ast }}^{{}}\left( x\right) v\left( x\right) \right) = \nabla_{x} \left( f_{X^{\ast }} \left( \bar x\right) v\left( \bar x\right) \right)
	\end{align*}
	since $\lim_{x \rightarrow -\infty} \nabla_x \left(f_{X^{\ast }} \left( x\right) v\left( x\right)\right) = 0$ because $v(x) = \tau^2 E[\xi^2|X_i^*=x]$, and $E[\xi^2|X_i^*=x]$ and $\nabla_x E[\xi^2|X_i^* = x]$ are bounded under Assumption \ref{ass:NPID:dominance}. The remainder takes the form
	\begin{align*}
		\int_{- \infty}^{\bar x} R_{f_{X^*}}(x; \tau) dx = \int_{- \infty}^{\bar x} \left\{\int \left(\int_{0}^\tau \varphi (x, u, t) dt\right) du\right\} dx,
	\end{align*}
	where
	\begin{align*}
		\varphi (x,u,t) \equiv \frac{(-1)^p}{(p-1)!} u^p \nabla_x^p \{f_{X^*}(x - tu) f_{\xi|X^*}(u|x - tu)\} (\tau - t)^{p-1}.
	\end{align*}
	Next, note that, using Assumption \ref{ass:NPID:integrable derivatives}, for all $u$ and $t \in [0, \tau]$
	\begin{align*}
		\int \abs{\varphi (x,u,t)} dx \leq C \abs{u}^p (\tau - t)^{p-1} \sup_{\tilde x \in \suppX}  \sum_{\ell = 0}^p \abs{\nabla_x^p f_{\xi|X^*} (u | \tilde x) },
	\end{align*}
	where $C > 0$ is a generic constant. Next, using Assumption \ref{ass:NPID:dominance},
	\begin{align*}
		\int \left( \int_{0}^\tau \abs{u}^p (\tau - t)^{p-1} \sup_{\tilde x \in \suppX}  \sum_{\ell = 0}^p \abs{\nabla_x^p f_{\xi|X^*} (u | \tilde x) } dt\right) du \leq C \tau^p.
	\end{align*}
	Hence, using Fubini–Tonelli's theorem, we conclude that
	\begin{align*}
		\underset{(-\infty,\bar x] \times \mathbb R \times [0,\tau]}{\iiint}  \abs{\phi (x,u,t)} dx du dt \leq C \tau^p,
	\end{align*}
	and the order of integration can be interchanged. Specifically, it also implies that
	\begin{align*}
		\int_{- \infty}^{\bar x} R_{f_{X^*}}(x; \tau) dx \leq C \tau^p.
	\end{align*}
	Hence, we conclude
	\begin{align*}
		F_{X}(\bar x) = F_{X^*}(\bar x) + \frac{1}{2} \nabla_x \left( f_{X^{\ast }}^{{}}\left( \bar x\right) v\left( \bar x\right)
		\right) + O(\tau^p), 
	\end{align*}
	so
	\begin{align}
		F_{X^*}(\bar x) &= F_{X}(\bar x) - \frac{1}{2} \nabla_x \left( f_{X^{\ast }}^{{}}\left( \bar x\right) v\left( \bar x\right)
		\right) + O(\tau^p) \notag \\
		&= F_{X}(\bar x) - \frac{1}{2} \left( f_{X^{\ast }}^{\prime}\left( \bar x\right) v\left( \bar x\right) + f_{X^{\ast }}\left( \bar x\right) v'\left( \bar x\right) \right) + O(\tau^p) \notag \\
		&= F_{X}(\bar x) - \frac{1}{2} \left( f_{X^{}}^{\prime}\left( \bar x\right) v\left( \bar x\right) + f_{X^{}}\left( \bar x\right) v'\left( \bar x\right) \right) + O(\tau^p) \notag \\
		&= F_{X}(\bar x) - \frac{1}{2} \nabla_x \left( f_{X}^{{}}\left( \bar x\right) v\left( \bar x\right)
		\right) + O(\tau^p), \label{eq:pf:CDF expansion}
	\end{align}
	where we used $f_X (\bar x) = f_{X^*} (\bar x) + O(\tau^2)$ and $f_X' (\bar x) = f_{X^*}' (\bar x) + O(\tau^2)$ (established the the proof of Lemma \ref{lem:smooth q prime and s prime}). This completes the proof of the first part. %

	\medskip

	Next, note that
	\begin{align*}
		Q_{X^*} (s) - Q_{X}(s) &= Q_{X^*} (s) - Q_{X^*} \left(F_{X^*} \left(Q_{X}(s)\right) \right)\\
			&= Q_{X^*}' (s) \left(F_X\left(Q_X(s)\right) - F_{X^*} \left(Q_{X}(s)\right) \right) + \frac{1}{2} Q_{X^*}'' (\tilde s) \left(F_X\left(Q_X(s)\right) - F_{X^*} \left(Q_{X}(s)\right) \right)^2,
	\end{align*}
	where $Q_{X^*}'(s) = \frac{1}{f_{X^*}' (Q_{X^*}(s))}$ and $Q_{X^*}''(s) = -\frac{f_{X^*}'' (Q_{X^*}(s))}{\left(f_{X^*}' (Q_{X^*}(s))\right)^3}$, and $\tilde s$ lies between $s = F_X(Q_X(s))$ and $F_{X^*} \left(Q_{X}(s)\right)$. Recall that \eqref{eq:pf:CDF expansion} implies
	\begin{align*}
		F_X(Q_x(s)) - F_{X^*}(Q_X(s)) &= \frac{1}{2} \nabla_x \left\{f_X(x) v(x) \right\}\big \vert_{x = Q_x(s)} + O(\tau^p) \\
		&= O (\tau^2).
	\end{align*}
	Hence, we conclude $Q_{X^*}(s) - Q_{X}(s) = O(\tau^2)$. This implies
	\begin{align*}
		\frac{1}{f_{X^*}' (Q_{X^*}(s))} &= \frac{1}{f_{X^*}' (Q_{X}(s))} + O(\tau^2) \\
		&= \frac{1}{f_{X}' (Q_{X}(s))} + O(\tau^2).
	\end{align*}
	Thus, we obtain
	\begin{align*}
		Q_{X^*} (s) - Q_{X}(s) &= \frac{1}{f_{X}'\left(Q_X(s)\right)} \; \frac{1}{2} \nabla_x \left\{f_X(x) v(x) \right\}\big \vert_{x = Q_x(s)}  + O(\tau^p) \\
		&=\frac{1}{2}\left\{s_{X}\left( Q_{X}\left( s\right) \right)  v\left( Q_{X}\left(s\right) \right) +\nabla _{x} v\left( Q_{X}\left( s\right) \right) \right\} + O(\tau^p).
	\end{align*}

\end{proof}

\subsubsection{Proof of Theorem \ref{thm:NCME rho tilde}}
\label{ssec:A:NCME ID proof}

\begin{proof}[Proof of Theorem \ref{thm:NCME rho tilde}]
	First, in the proof of Theorem \ref{thm:NPID-non-cl:IV}, we show that
	\begin{align*}
		\tilde v\left(Q_{X^* } \left(F_{\mathcal X^*} (\varkappa)\right)\right) &= v\left(Q_{X^* } \left(F_{\mathcal X^*} (\varkappa)\right)\right) + O (\tau^p),\\
		\tilde v'\left(Q_{X^* } \left(F_{\mathcal X^*} (\varkappa)\right)\right) &= v'\left(Q_{X^* } \left(F_{\mathcal X^*} (\varkappa)\right)\right) + O (\tau^p).
	\end{align*}
	Combining the above with the result of Lemma \ref{lem:CDF and QF}, we establish
	\begin{align}
		\label{eq:pf:delta Q star}		
		\tilde Q_{X^*} (F_{\mathcal X^*} (\varkappa)) - Q_{X^*} \left(F_{\mathcal X^*} (\varkappa)\right) = O(\tau^p).
	\end{align}
	Since $\rho'(\cdot)$ is bounded, we also have
	\begin{align}
		\label{eq:pf:delta rho Q star}
		\abs{\rho (\tilde Q_{X^*} (F_{\mathcal X^*} (\varkappa))) - \rho \left(Q_{X^*} \left(F_{\mathcal X^*} (\varkappa)\right)\right) } = O(\tau^p).
	\end{align}

	Second, there exists $\delta > 0$ such that $f_{X^*|Z}(x|z_1)$ and $f_{X^*}(x|z_2)$ are bounded away from zero for $x \in B_{\delta} \left(Q_{X^*} \left(F_{\mathcal X^*} (\varkappa)\right) \right)$. Since all the remainders such as $O(\cdot)$ and $O_x(\cdot)$ are explicitly bounded in the proof of Theorem \ref{thm:NPID-non-cl:IV} and those bounds are uniform in $x \in \mathcal B_\delta (Q_{X^*}(F_{\mathcal X^*}(\varkappa)))$, we have
	\begin{align}
		\label{eq:pf:uniform rate for rho}
		\sup_{x \in B_{\delta} \left(Q_{X^*} \left(F_{\mathcal X^*} (\varkappa)\right) \right) } \abs{\tilde \rho_{X^*} (x,z_1) - \rho_{X^*} (x)} = O(\tau^p).
	\end{align}

	Finally,
	\begin{align*}
		\abs{\tilde \rho_{\mathcal X^*} (\varkappa,z_1) - \rho_{\mathcal X^*} (\varkappa)} = &\abs{\tilde \rho_{X^*} (\tilde Q_{X^*} \left(F_{\mathcal X^*} (\varkappa)\right),z_1) -  \rho_{X^*} \left(Q_{X^*} \left(F_{\mathcal X^*} (\varkappa)\right)\right)}\\ %
		\leq &\sup_{x \in B_{\delta} \left(Q_{X^*} \left(F_{\mathcal X^*} (\varkappa)\right) \right) } \abs{\tilde \rho_{X^*} (x,z_1) - \rho_{X^*} (x)} \\ &+ \abs{\rho (\tilde Q_{X^*} (F_{\mathcal X^*} (\varkappa))) - \rho \left(Q_{X^*} \left(F_{\mathcal X^*} (\varkappa)\right)\right) },
	\end{align*}
	where the inequality holds for sufficiently small values of $\tau$ due to \eqref{eq:pf:delta Q star}. %
	Combining \eqref{eq:pf:delta rho Q star} and \eqref{eq:pf:uniform rate for rho} together completes the proof. %
\end{proof}

\section{Verification of Assumption \ref{ass:NPID:dominance}}
\label{sec: verification of dominance}

\begin{lemma}
	\label{lem: dominance verification}
	Suppose $\xi = \sigma(X^*) \zeta$, where $\zeta$ is independent of $X^*$. Also, suppose that the following conditions are satisfied:
	\begin{enumerate}[(i)]
		\item \label{item:dom_ver:smoothness} $\sigma(\cdot)$ and $f_\zeta(\cdot)$ are bounded (non-negative) functions with $m$ bounded derivatives, and $\sigma(\cdot)$ is bounded away from 0;

		\item \label{item:dom_ver:monotonicity} for sufficiently large values of $t$, $t^{m+k+1} \vert f_\zeta^{(k)}(t) \vert$ and $t^{m+k+1} \vert f_\zeta^{(k)}(-t) \vert$, for $k \in \{0, \ldots, m\}$, are decreasing functions;

		\item \label{item:dom_ver:moments}$ \int \vert \zeta \vert^{m+k} \left \vert f_\zeta^{(k)} \left(\zeta\right) \right \vert d\zeta < C$ for $ k \in \{0,\ldots, m\}$.
	\end{enumerate}
	Then, Assumption \ref{ass:NPID:dominance} is satisfied.
\end{lemma}

\begin{proof}[Proof of Lemma \ref{lem: dominance verification}]
	First, notice that $f_{\xi|X^*} (u|x) = \frac{1}{\sigma(x)} f_{\zeta} \left(\frac{u}{\sigma(x)}\right)$.

	We want to show that
	\begin{align*}
		\int \vert u \vert^m \sup_{\tilde x \in \suppX} \left \vert \nabla_x^\ell f_{\xi|X^*} (u| \tilde x) \right \vert du < C
	\end{align*}
	for $\ell \in \{0, \ldots, m\}$, where $C > 0$ is a universal constant. Since the derivatives of $f_{\xi|X^*}(u|x)$ are uniformly bounded (by condition \eqref{item:dom_ver:smoothness}), we just focus on showing
	\begin{align*}
		\int_{\underline u}^\infty \vert u \vert^m \sup_{\tilde x \in \suppX} \left \vert \nabla_x^\ell f_{\xi|X^*} (u| \tilde x) \right \vert du < C
	\end{align*}
	for some (generic) $\underline u > 0$ ($\int_{-\infty}^{-\underline u} \vert u \vert^m \sup_{\tilde x \in \suppX} \left \vert \nabla_x^\ell f_{\xi|X^*} (u| \tilde x) \right \vert du$ can be bounded using a similar argument).	

	Starting with $\ell = 0$, we verify
	\begin{align*}
		\int_{\underline u}^\infty \vert u \vert^m \sup_{\tilde x \in \suppX} f_{\xi|X^*} (u| \tilde x) du < C.
	\end{align*}
	Notice that
	\begin{align*}
		\sup_{x \in \suppX} f_{\xi|X^*} (u|x) = \sup_{s \in \Sigma} \left\{ \frac{1}{s} f_\zeta \left(\frac{u}{s}\right) \right\},
	\end{align*}
	where $\Sigma = \{\sigma(x): x \in \suppX\}$ is bounded. Let $s^* = \sup_{x \in \suppX} \sigma(x)$. Condition \eqref{item:dom_ver:monotonicity} implies that there exists $\underline u > 0$ such that for all $\vert u \vert > \underline u$ we have
	\begin{align*}
		\sup_{s \in \Sigma} \left\{ \frac{1}{s} f_\zeta \left(\frac{u}{s}\right) \right\} = \frac{1}{s^*} f_{\zeta} \left(\frac{u}{s^*}\right).
	\end{align*}
	Then,
	\begin{align*}
		\int_{\underline u}^\infty \vert u \vert^m \sup_{\tilde x \in \suppX} f_{\xi|X^*} (u| \tilde x) du = \int_{\underline u}^\infty \vert u \vert^m  \frac{1}{s^*} f_{\zeta} \left(\frac{u}{s^*}\right) du < C,
	\end{align*}
	where the last inequality is due to condition \eqref{item:dom_ver:moments}.

	Next, we verify that
	\begin{align*}
		\int_{\underline u}^\infty \vert u \vert^m \sup_{\tilde x \in \suppX} \left\vert \nabla_x^{\ell} f_{\xi|X^*} (u| \tilde x) \right\vert  du < C
	\end{align*}
	for $\ell \in \{1 , \ldots, m\}$.
	Since the derivatives of $\sigma(\cdot)$ are bounded, we have
	\begin{align*}
		\sup_{\tilde x \in \suppX} \left\vert \nabla_x^{\ell} f_{\xi|X^*} (u| \tilde x) \right\vert \leqslant C \sum_{k=1}^\ell \sup_{s \in \Sigma} \left \vert \nabla_s^{k}  \left\{ \frac{1}{s} f_\zeta \left(\frac{u}{s}\right) \right\} \right \vert,
	\end{align*}
	so it is sufficient to verify
	\begin{align*}
		\int_{\underline u}^{\infty} \vert u\vert^m \sup_{s \in \Sigma} \left \vert \nabla_s^{\ell}  \left\{ \frac{1}{s} f_\zeta \left(\frac{u}{s}\right) \right\} \right \vert du < C,
	\end{align*}
	for $\ell \in \{1, \ldots, m\}$. Next,
	\begin{align*}
		\nabla_s^{\ell}  \left\{ \frac{1}{s} f_\zeta \left(\frac{u}{s}\right) \right\} = \sum_{k=0}^{\ell} a_{\ell k} \frac{u^k}{s^{\ell + k + 1}} f^{(k)} \left(\frac{u}{s}\right)
	\end{align*}
	for some constants $a_{\ell k}$. Then, condition \eqref{item:dom_ver:monotonicity} implies that there exists $\underline u > 0$ such that for all $\vert u \vert > \underline u$ we have
	\begin{align*}
		\sup_{s \in \Sigma} \left \vert \nabla_s^{\ell}  \left\{ \frac{1}{s} f_\zeta \left(\frac{u}{s}\right) \right\} \right \vert = \left \vert \nabla_s^{\ell}  \left\{ \frac{1}{s} f_\zeta \left(\frac{u}{s}\right) \right\} \right \vert_{s = s^*} = \left \vert \sum_{k=0}^{\ell} a_{\ell k} \frac{u^k}{(s^*)^{\ell + k + 1}} f^{(k)} \left(\frac{u}{s^*}\right) \right \vert.
	\end{align*}
	Then,
	\begin{align*}
		\int_{\underline u}^\infty \vert u \vert^m \sup_{\tilde x \in \suppX} \left\vert \nabla_x^{\ell} f_{\xi|X^*} (u| \tilde x) \right\vert  du &< C \int_{\underline u}^\infty \sum_{k=0}^\ell \frac{\vert u\vert^{m+k}}{(s^*)^{\ell + k + 1}} \left \vert f_\zeta^{(k)} \left(\frac{u}{s^*}\right) \right \vert du \\ &< C \sum_{k=0}^\ell \int \vert \zeta \vert^{m+k} \left \vert f_\zeta^{(k)} \left(\zeta\right) \right \vert d\zeta \\
		&<C,
	\end{align*}
	where the last inequality is due to condition \eqref{item:dom_ver:moments}.
\end{proof}

\end{appendices}%

\end{document}